%% file: ArxivVer1.tex
\documentclass{llncs}

\usepackage{pslatex,amssymb,pgf,tikz,lscape,amsmath,cite,mathrsfs,latexsym}
\usetikzlibrary{arrows,automata}

\input{macros}

\begin{document}
\title{The Unary Fragments of Metric Interval Temporal Logic: Bounded versus Lower bound Constraints\\ (Full Version)}
\author{Paritosh K.~Pandya and Simoni S.~Shah}
\institute{Tata Institute of Fundamental Research,
Colaba, Mumbai \textit{400005}, India}
\pagestyle{empty} 

\maketitle

\begin{abstract}
\footnote{This is the full version of the paper of the same name presented at ATVA, 2012 (doi: 10.1007/978-3-642-33386-6)}
We study two unary fragments of the well-known metric interval temporal logic $\mitlus$ that was
originally proposed by Alur and Henzinger, and we pin down their expressiveness
as well as satisfaction complexities. We show that $\mitlfpinf$ which has unary
modalities with only lower-bound constraints is (surprisingly) expressively complete for Partially Ordered 2-Way Deterministic Timed Automata ($\potdta$) and the reduction from logic to automaton gives us its NP-complete satisfiability. We also show that the fragment $\mitlfpb$ having unary modalities with only bounded intervals has \nexptime-complete satisfiability. But strangely, 
\mitlfpb\/ is strictly less expressive than \mitlfpinf. 
We provide a comprehensive picture of the decidability and expressiveness of various unary fragments of \mitl.
\end{abstract}

\section{Introduction}
Temporal logics are a well known notation for specifying properties of reactive systems.
Reductions between temporal logics and finite state automata have been very influential in formulating decision procedures and model checking of temporal logic properties. 
However, extending this paradigm to real-time logics and timed automata has been challenging.

Metric Temporal Logic $\mtlus$ is a well established logic for specifying quantitative properties of  timed behaviors in real-time. 
In this logic, the temporal modalities $\until_I$ and $\since_I$ are time constrained by a time interval $I$. 
A formula $\phi \until_I \psi$ holds at a position $i$ provided
there exists a strictly later position $j$ where $\psi$ holds and $\phi$ must hold for all in between positions. Moreover the ``time distance'' between $j$ and $i$ must be in the interval $I$. 
Interval $I = \langle l,u \rangle$ has integer valued endpoints  and it can be open, closed, half open, or singular (i.e. $[c,c]$). It can even be unbounded, i.e. of the form $\langle l,\infty)$.  Unary modalities $\fut_I \phi$ and $\past_I \phi$ can be defined as $(true) \until_I \phi$ and  $(true) \since_I \phi$, respectively. 
Unfortunately, satisfiability of $\mtlus$ formulae and their model checking (against timed automata) are both undecidable in general \cite{AH93,Hen91}.

In their seminal paper \cite{AFH96}, the authors proposed the sub logic $\mitlus$ having only non-punctual (or non-singular) intervals.  Alur and Henzinger \cite{AFH96,AH92} showed that the logic \mitlus\/ has \expspace-complete satisfiability\footnote{This assumes that the time constants occurring in the formula are written in binary. We follow the same convention throughout this paper.}. 
In another significant paper \cite{BMOW08}, Bouyer {\em et al} showed that sublogic of $\mtlus$ with only bounded intervals, denoted $\mtlusb$, also has \expspace-complete satisfiability. These results are practically significant since many real time properties can be stated with bounded or non-punctual interval constraints.

In quest for more efficiently decidable timed logics, Alur and Henzinger considered the fragment \mitluszinf\/ consisting only of ``one-sided'' intervals, and showed that it  has \pspace-complete satisfiability. Here, allowed intervals are of the form $[0,u\rangle$   or $\langle l,\infty)$ thereby enforcing either an upper bound or a lower bound time constraint in each modality.  

Several real-time properties of systems may be specified by using the unary \emph{future} and \emph{past} modalities alone. In the untimed case of finite words, the unary fragment of logic \ltlus\/ has a  special position: the  unary temporal logic $\ltlfp$ has NP-complete satisfiability 
\cite{EVW02} and it expresses exactly the unambiguous star-free languages which are characterized by Partially ordered 2-Way Deterministic Finite Automata (\potdfa) \cite{STV01}. 
On the other hand, the \pspace-complete satisfiability of $\ltlus$  drops to \np-complete satisfiability for unary temporal logic $\ltlfp$\/ \cite{EVW02}.
Automata based characterizations for the above two logics are also well known: \ltlus\/- definable languages are exactly the star-free regular languages  which are characterized by counter-free automata, where as \ltlfp\/- definable languages exactly correspond to the unambiguous star-free languages \cite{TT02} which are characterized by Partially ordered 2-Way Deterministic Automata (\potdfa) \cite{STV01}.

Inspired by the above, in this paper, we investigate several ``unary'' fragments of \mitlus\/ and we pin down their exact decision complexities as well as expressive powers.
\oomit{
Metric temporal logic can be defined over various classes of time frames, such as point-wise or continuous time. Moreover, the behaviors may be infinite or finite. The decidability and expressiveness of $\mtl$ changes very significantly with the nature of time \cite{OW07,DP07}. }
\emph{In this paper, we confine ourselves to point-wise $\mitl$ with finite strictly monotonic time}, i.e. the models are finite timed words where no two letters have the same time stamp. 

As our main results, we identify two fragments of unary logic $\mitlfp$ for which a remarkable drop in complexity of checking satisfiability is observed, and we study their automata as well as expressive powers. These fragments are as follows.
\begin{itemize} 
\item  Logic $\mitlfpinf$ embodying only unary ``lower-bound'' interval constraints of the form $F_{\langle l,\infty)}$ and $P_{\langle l,\infty)}$. We show that satisfiability of this logic is \np-complete.
\item Logic $\mitlfpb$ having only unary modalities $F_{\langle l,u\rangle}$ and $P_{\langle l,u\rangle}$ with bounded and non-singular interval constraints  where ($u \not= \infty$).
We show that satisfiability of this logic is \nexptime-complete.
\end{itemize}
In both cases, an automata theoretic decision procedure is given as a language preserving reduction from the logic  to  Partially Ordered 2-Way Deterministic Timed Automata (\potdta). These automata are a subclass of the 
2Way Deterministic Timed Automata $\twodta$ of Alur and Henzinger \cite{AH92} and they incorporate the notion of partial-ordering of states. 
They define a subclass of timed regular languages called unambiguous timed regular languages (\tul) (see \cite{PS10}).
\potdta\/ have several attractive features: they are boolean closed (with linear blowup only) and  their non-emptiness checking is \np-complete.
\oomit{ The
\potdta\/ have several attractive features: they are boolean closed and each boolean operation gives rise to linear blowup in their size.
Moreover,  their non-emptiness is NP-complete and language containment is \conp-complete  \cite{PS10}.
}
The  properties of \potdta\/ together with our reductions give the requisite decision procedures for satisfiability checking of logics \mitlfpinf\/ and \mitlfpb. \oomit{Both the reductions rely upon a novel scheme of  clocking (freezing) the times of first and last occurrences of subformulas in a pre-determined finite set of time intervals  to evaluate the truth of the formula. This can be done in  an inductive, bottom up fashion in successive passes of the two way automaton. 
}

The reduction from $\mitlfpinf$ to \potdta\/ uses a nice optimization which becomes possible in this sublogic: truth of a formula at any point can be determined as a simple condition between times of first and last occurrences of its modal subformulas and current time. A much more sophisticated but related optimization is required for the logic $\mitlfpb$ with both upper and lower bound constraints: truth of a formula at any point in a unit interval can be related to the times of first and last occurrences  of its immediate modal subformulas in some ``related'' unit intervals. The result is an inductive bottom up evaluation of the first and last occurrences of subformulas which is carried out in successive passes of the two way deterministic timed automaton.

For both the logics, we show that  our decision procedures are optimal.
\oomit{ the hardness of these and other related logics is established by reduction of a suitable tiling problem to the satisfiability of the concerned unary \mitl\/ fragment.}  
We also verify that the logic $\mitlf$ consisting only of the unary future fragment of $\mitlus$  already exhibits \expspace-complete satisfiability. Moreover, the unary future fragment $\mitlfz$ with only upper bound constraints has \pspace-complete satisfiability, whereas  $\mitlfpinf$ with only lower bound constraints has \np-complete satisfiability.
A comprehensive picture of decision complexities of fragments of $\mitlfp$ is obtained and summarized in Figure \ref{fig:complexity}.

We also study the expressive powers of logics \mitlfpinf\/ and \mitlfpb.
We establish that \mitlfpinf\/  is expressively complete for \potdta, and hence it can define all unambiguous timed regular languages (\tul). 
This is quite  surprising as $\potdta$ include guards with simultaneous upper and lower bound constraints as well as punctual constraints, albeit only occurring deterministically. Expressing these in $\mitlfpinf$, which has only lower bound constraints, is tricky.
\oomit{This \ttlxy\/ logic embodies the freeze quantification (of \tptlus)  and it has both punctual and unbounded constraints, albeit occurring only within deterministic modalities. 
We now show that  $\mitlfpinf$ can express all $\ttlxy$ formulas, and the two logics are expressively equivalent. The reduction from $\ttlxy$ to \mitlfpinf\/ is a variant of another such encoding into the more expressive logic \mitlfp\/ given earlier (see \cite{PS11}).}
We remark that $\mitlfpinf \equiv \potdta$ is a rare instance of a precise logic automaton connection within the \mtlus\/ family of timed logics.

We also establish that  $\mitlfpinf$ is strictly more expressive than the bounded unary logic $\mitlfpb$. Combining these results with decision complexities, we conclude that $\mitlfpb$, although less expressive, is exponentially more succinct as compared to the logic $\mitlfpinf$. Completing the picture, we show that, for expressiveness, $\mitlfpb \subsetneq \mitlfpinf \subsetneq \mitlfpzinf \subsetneq \mitlfp$. For each logic, we give a sample property that cannot be expressed in the contained logic (see Figure \ref{fig:unaryexpress}). The inexpressibility of these properties in lower logics are proved using an EF theorem for $\mtl$ formulated earlier \cite{PS11}.

For logic $\mitlfpb$, the reduction relies on the property that checking truth of a unary modal formula $M_{\langle l,l+1 \rangle} \phi$ at any position $T$ of a given unit interval $[r,r+1)$ can be formulated as simple condition over $T$ and the times of
first and last occurrences of $\phi$ in some related unit intervals (such as
$[l+r,l+r+1)$. We call this the horizontal stacking of unit intervals
Some remarks on our reductions are appropriate here.
It should be noted that these logics have both future and past modalities and these naturally translate to the two-wayness of the automata. An important feature of our reduction is that checking of satisfiability of a modal subformula $F_I \phi$ reduces searching for ``last'' occurrence of  $\phi$ within some specified subintervals, and remembering its time stamp. This can be carried out by one backward scan of the automaton.  Similarly, for the past formula $P_i \phi$ we need a forward scan.

\begin{figure}
\begin{tikzpicture}[scale=0.9,transform shape]
\draw(0,0) -- (14,0); \draw(0,1.5) -- (14,1.5);\draw(0,3) -- (14,3);\draw(0,4.5) -- (14,4.5);\draw(0,6) -- (14,6);
\draw [dashed] (3,0) -- (3,6);
\draw (0.5,0.7) node{\np-complete}; \draw (1,2.2) node{\pspace-complete}; \draw(1,3.5)node{\nexp-complete}; \draw (1,5.2) node{\expspace-complete};

\draw (9,1) node (A) [rectangle,draw] {\mitlfpinf};
\draw (12,0.5) node (L) [rectangle] {\mitlfinf};
\draw (5,1) node (B) [rectangle] {\ttl};
\draw (12,1.9) node (C) [rectangle] {\mitlfz};
\draw (9,2.3) node (K) [rectangle] {\mitlfpzinf};
\draw (5,2.7) node (D) [rectangle] {\mitluszinf};
\draw (12.5,3.5) node (E) [rectangle] {\mitlfb};
\draw (8,3.8) node (F) [rectangle,draw] {\mitlfpb};
\draw (6.7,4.8) node (G) [rectangle] {\mtlusb};
\draw (12,4.8) node (H) [rectangle] {\mitlf};
\draw (9,5.3) node (I) [rectangle] {\mitlfp};
\draw (4.3,5.8) node (J) [rectangle] {\mitlus};
\draw (5,1.8) node (M) [rectangle] {\mitlusinf};

\draw [->](J)-- (I);
\draw [->] (I)-- (H);
\draw [->](H)-- (E);
\draw [->](I)-- (F);
\draw [->](G)-- (F);
\draw [->](F)-- (E);
\draw [->](K)-- (A);
\draw [->](D)-- (K);
\draw [->](K)-- (C);
\draw [->](A)-- (L);
\draw [->](D)-- (M);
\draw [->](M)-- (A);
\draw [->](J)--(D);
\end{tikzpicture}
\caption{Unary MITL: fragments with satisfiability complexities. Arrows indicate syntactic inclusion. The boxed logics are the two main fragments studied in this chapter.}
\label{fig:complexity}
\end{figure}
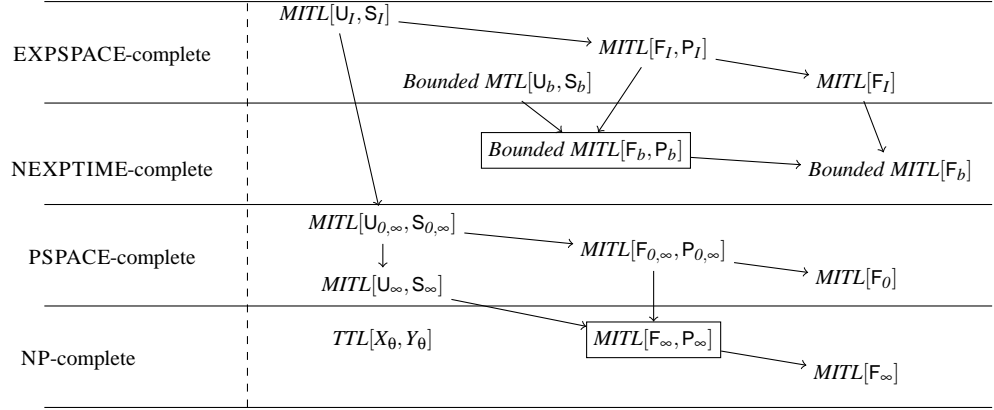

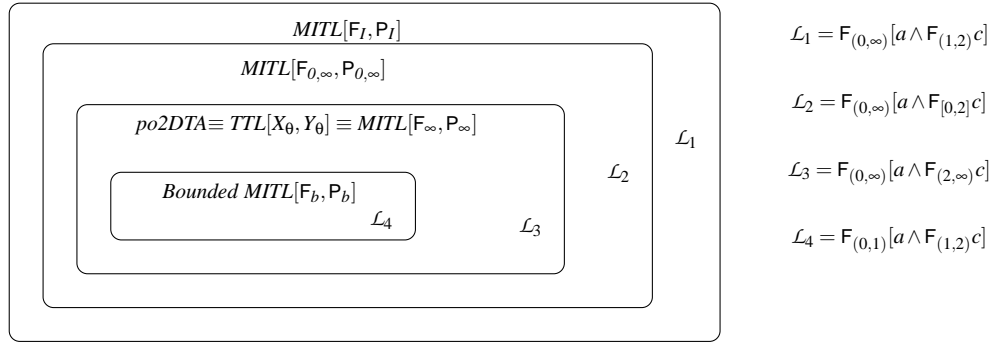
\begin{figure}
\begin{tikzpicture}[scale = 0.9, transform shape]
 \draw [rounded corners] (0,0) rectangle (10.5,5);
\draw (5,4.6) node {\mitlfp};
\draw (10,3) node {$\mathcal L_1$};

 \draw [rounded corners] (0.5,0.5) rectangle (9.5,4.4);
\draw (4.5,4) node {\mitlfpzinf};
\draw (9,2.5) node {$\mathcal L_2$};

 \draw [rounded corners] (1,1) rectangle (8.2,3.5);
\draw (4.4,3.2) node {\potdta $\equiv$ \ttl\/ $\equiv$ \mitlfpinf};
\draw (7.7,1.7) node {$\mathcal L_3$};

\draw [rounded corners] (1.5,1.5) rectangle (6,2.5);
\draw (3.7,2.2) node{\mitlfpb};
\draw (5.5,1.8) node {$\mathcal L_4$};

\draw (13,4.5) node {$\mathcal L_1 = \fut_{(0,\infty)}[a\land \fut_{(1,2)}c]$};
\draw (13,3.5) node {$\mathcal L_2 = \fut_{(0,\infty)}[a\land \fut_{[0,2]}c]$};
\draw (13,2.5) node {$\mathcal L_3 = \fut_{(0,\infty)}[a\land \fut_{(2,\infty)}c]$};
\draw (13,1.5) node {$\mathcal L_4 = \fut_{(0,1)}[a\land \fut_{(1,2)}c]$};

\end{tikzpicture}
\caption{Expressiveness of Unary \mitl\/ fragments}
\label{fig:unaryexpress}
\end{figure}

\section{Unary \mitl\/ and its fragments}
\label{sec:mitlsem}
\begin{definition}{[Timed Words]}
A finite timed word over an alphabet $\Sigma$ is a finite sequence $\rho = (\sigma_1,\tau_1), \cdots (\sigma_n,\tau_n)$, of event-time stamp pairs 
such that $\forall i ~.~ \sigma_i\in \Sigma$ and the 
sequence of time stamps is non-decreasing: $\forall i<n ~.~ \tau_i\leq \tau_{i+1}$. This gives weakly monotonic timed words. If
time stamps are strictly increasing, i.e. $\forall i<n ~.~ \tau_i< \tau_{i+1}$, the timed word is strictly monotonic.
\end{definition}
The length of $\rho$ is denoted by $\#\rho$, and $dom(\rho) =
\{1,...\#\rho\}$. For convenience, we assume  that $\tau_1 = 0$ as
this simplifies the treatment of the initial semantics of timed logics. The timed word $\rho$ can alternately be represented as 
$\rho=(\overline{\sigma},\overline{\tau})$ with $\overline{\sigma} = \sigma_1, \cdots, \sigma_n$ and 
$\overline{\tau} = \tau_1, \cdots,\tau_n$. Let
$untime(\rho)=\overline{\sigma}$ be the untimed word of $\rho$ and $alph(\rho)\subseteq\Sigma$ be the set of events that occur in $\rho$. Let $\rho(i)...\rho(j)$ for some $1\leq~ i~\leq~j ~\leq \#\rho$ be the factor of $\rho$ given by $(\sigma_i,\tau_i)\cdots(\sigma_j,\tau_j)$.
Let $T\Sigma^*$ be the set of timed words over the alphabet $\Sigma$.

The logic MTL \cite{Koy90,AH91} extends Linear Temporal Logic by adding timing
constraints to the "Until" and "Since" modalities of LTL, using timed intervals. We consider the unary fragment of this logic called \mtlfp. Let $I$ range over the set of intervals with non-negative integers as end-points. The syntax of \mtlfp\/ is as follows:
\[
\phi ::= a ~\mid~ \phi \lor\phi ~\mid~ \neg \phi ~\mid~ \fut_I\phi ~\mid~ \past_I\phi
\]

\begin{remark}
In this paper, we study \mtl\/ with interval constraints given by timed intervals with integer end-points. In literature, \mtl\/ with interval constraints with rational end-points are often considered. However, it is important to note that properties expressed by the latter may also be expressed by the former, by scaling the intervals as well as the timestamps in the timed word models appropriately.
\end{remark}

Let $\rho=(\overline{\sigma},\overline{\tau})$ be a timed word and let
$i\in dom(\rho)$. The semantics of \mtlfp\/ formulas over pointwise models is as below:
\[
\begin{array}{rcl}
\rho,i\models a & \fif & \sigma_i=a\\
\rho,i\models \neg \phi & \fif & \rho,i\not\models \phi\\
\rho,i\models \phi_1\lor \phi_2 & \fif & \rho,i\models\phi_1 ~\mbox{or} ~\rho,i\models\phi_2\\
\rho,i\models \fut_I \phi_1 & \fif & \exists j>i .~ \rho,j\models\phi_1\\
\rho,i\models \phi_1 \past_I \phi_2 & \fif & \exists j<i ~.~ \rho,j\models\phi_1\\
\end{array}
\]
The language of an \mtlfp\/ formula $\phi$ is given by $\mathcal L(\phi) = 
\{\rho ~\mid~ \rho, 1 \models \phi\}$.\\

\mitlfp\/ is the fragment of \mtlfp\/ which allows only non-punctual intervals to constrain the $\fut$ and $\past$ modalities. Some fragments of \mitlfp\/ that we shall consider in this paper are as follows. See Figure \ref{fig:unaryexpress} for examples.
\begin{itemize}
\item \mitlfpzinf allows only interval constraints of the form $[0,u\rangle$ or $\langle l,\infty)$. Thus, each modality enforces either an upper bound or a lower bound constraint.  
\item \mitlfpb\/ is \mitlfp\/ with the added restriction that all interval constraints are bounded intervals of the form $\langle l,u \rangle $ with $u \not=\infty$. 
\item \mitlfpinf\/ is the fragment of \mitlfp\/ where all interval constraints are ``lower bound'' constraints of the form $\langle l, \infty)$. 
\item \mitlfpz\/ is the fragment in which all interval constraints (whether bounded or unbounded) are ``upper bound'' constraints of the form $[0,u \rangle$.
\item \mitlf, \mitlfzinf, \mitlfb, \mitlfinf\/ and \mitlfz\/ are the corresponding \emph{future}-only fragments. 
\end{itemize}

\paragraph{Size of \mitlfp\/ formulas}
Consider any \mitlfp\/ formula $\phi$, represented as a DAG. Let $n$ be the number of modal operators in the DAG of $\phi$. Let $k$ be the product of all constants that occur in $\phi$. Then the \emph{modal-DAG size} $l$ of $\phi$, whose constants are presented in some logarithmic encoding (e.g. binary) is within constant factors of $(n + log k)$. 


\begin{definition}{[Normal Form for \mitlfp]}
\label{def:norm}
Let ${\mathscr B}(\{\psi_i\})$ denote a {\em boolean} combination of formulas from the finite set $\{ \psi_i \}$. Then a normal form formula $\phi$ is given by 
\[
 \begin{array}{l}
   \phi ~~:=~~ \bigvee\limits_{a \in \Sigma} ~ (a \land {\mathscr B}(\{\psi_i\})) \\
   \mbox{where each $\psi_i$ is a modal formula of the form}\\
   \psi ~~:=~~ F_{I}(\phi) ~\mid~ P_{I}(\phi)\\
   \mbox{where each $\phi$ is also in normal form.}
 \end{array}
\]

\end{definition}
A subformula $\phi$ in normal form is said to be an $\fut$-type \emph{modal argument (or modarg in brief)} if it occurs within an $\fut$-modality (as $\fut_I(\phi)$).  It is a $\past$-type modarg if it occurs as $\past_{I}(\phi)$. Each $\psi_i$ is said to be a \emph{modal sub formula}.

\begin{proposition}\label{prop:norm}
Every \mitlfp\/ formula $\zeta$ may be expressed as an equivalent normal form formula $\phi$ of modal-DAG size linear in the  modal-DAG size of $\zeta$.
\end{proposition}
\begin{proof}
Given $\zeta$, consider the equivalent formula $\zeta \land \bigvee\limits_{a \in \Sigma} a$. Transform this formula in disjunctive normal form treating modal subformulas as atomic. Now apply reductions such as 
$a\land b\land \mathscr B(\psi_i) ~\equiv~ \bot$ (if $a\neq b$) and $a\land \mathscr B(\psi_i)$ otherwise. The resulting formula is equivalent to $\zeta$.
Note that  DNF representation does not increase the modal-DAG size of the formula. Apply the same reduction to modargs recursively.  
\end{proof}

\subsection{\potdta}
In \cite{PS10}, we defined a special class of \twodta\/ called Partially-ordered 2-way Deterministic Timed Automata (\potdta). The only loops allowed in the transition graph of these automata are self-loops. This condition naturally defines a partial order on the set of states (hence the name). Another restriction is that clock resets may occur only on progress edges. THese are a useful class of automata for the following reasons: 
\begin{itemize}
\item The ``two-way'' nature of the automata naturally allows the simultaneous treatment of \textit{future} and \textit{past} modalities in timed temporal logics.
\item Since they are deterministic, complementation may be achieved trivially. In fact, the deterministic and two-way nature of the automata allow for boolean operations to be achieved with only a linear blow-up in the size of the automaton.
\item The size of the small model of a \potdta\/ is polynomial in the size of the automaton. Hence, language emptiness of a \potdta\/ is decidable with \np-complete complexity.
\end{itemize} 

\potdta\/ are formally defined below.

Let $C$ be a finite set of clocks. A \emph{guard} $g$ is a timing constraint on the clock values and has the form:\\
\tab  $  g ~:=~ \top ~\mid~ g_1 \land g_2 ~\mid~ x-T\approx c ~\mid~ T-x\approx c ~\mbox{where $\approx\in \{<,\leq, >,\geq,=\}$
  and $c\in\mathbb{N}$}$.\footnote{Note that the guards $x-T\approx c$ and $T-x\approx c$ implicitly include the conditions $x-T\geq 0$ and $T-x\geq 0$ respectively.}\\

Here, $T$ denotes the current time value. Let $G_C$ be the set of all guards over $C$. A clock valuation is a function which assigns to each clock a non-negative real number. Let $\nu,\tau\models g$ denote that a valuation $\nu$ satisfies the guard $g$ when $T$ is assigned a real value $\tau$. If $\nu$ is a clock valuation  and $x\in C$, let $\nu'=\nu\unite (x\to \tau)$ denote a valuation such that $\forall y\in C ~.~ y\neq x\Rightarrow \nu'(y)=\nu(y)$ and $\nu'(x)=\tau$. Two guards $g_1$ and $g_2$ are said to be disjoint if for all valuations $\nu$ and all reals $r$, 
we have $\nu,r\models \neg(g_1 \land g_2)$. A special valuation $\nu_{init}$ maps all clocks to 0.

Two-way automata ``detect'' the ends of a word, by appending the word with special \emph{end-markers} on either side. Hence, if $\rho = (\sigma_1,\tau_1) ...(\sigma_n,\tau_n)$ then the run of a \potdta\/ is defined on a timed word $\rho' = (\lend,0) (\sigma_1,\tau_1)... (\sigma_n,\tau_n),(\rend,\tau_n)$.

\begin{definition}[Syntax of \potdta]
Fix an alphabet $\Sigma$ and let $\Sigma' = \Sigma \cup \{\lend,\rend\}$. Let $C$ be a finite set of clocks. A po2DTA over alphabet $\Sigma$ is a tuple $M = (Q,\leq,\delta,s,t,r,C)$ where $(Q,\leq)$ is a partially
ordered and finite set of states such that $r,t$ are the only minimal elements and $s$ is the only maximal  element. Here, $s$ is the initial state, $t$ the accept state and $r$ the reject state. The set $Q \setminus \{t,r\}$ is partitioned into $Q_L$ and $Q_R$ (states which are \emph{entered from} the left and right respectively).  
The progress transition function is a partial function $\delta: ((Q_L \cup Q_R)\times \Sigma' \times G_C) \to Q \times 2^C)$
which specifies the progress transitions of the automaton, such that if $\delta(q,a,g) = (q',R)$ then $q' ~<~ q$ and $R \in 2^C$ is the subset of clocks that is reset to the current time stamp.
Every state $q$ in $Q \setminus
\{t,r\}$ has a default \textit{``else''} self-loop transition which is taken in all
such configurations for which no progress transition is enabled. 
Hence, the automaton continues to loop in a given state $q$ and scan the timed word in a single direction (right or left, depending on whether $q\in Q_L$ or $Q_R$ respectively), until one of the progress transitions is taken. Note that there are no transitions from the terminal states $r$ and $t$. 

\end{definition}

\begin{definition}[Run]
Let $\rho = (\sigma_1,\tau_1),(\sigma_2,\tau_2)... (\sigma_m,\tau_m)$ be a given timed word. 
The configuration of a po2DTA at any instant is given by $(q,\nu,l)$ where $q$ is the current state, the current value of the clocks is given by the valuation function $\nu$ and the current head 
position is $l \in dom(\rho')$. In this configuration, the head reads the letter $\sigma_l$ and the time stamp  $\tau_l$.

The run of a po2DTA on the timed word $\rho$ with and starting head position $k\in dom(\rho')$ and starting 
valuation $\nu$ is the (unique) sequence of configurations $(q_1,\nu_1,l_1)$ $\cdots$ $(q_n,\nu_n,l_n)$ such  that
\begin{itemize}
 \item Initialization: $q_1 = s$, $l_1=k$ and $\nu_1 = \nu$. The automaton always starts in the initial
 state $s$.
\item  If the automaton is in a configuration $(q_i,\nu_i,l_i)$ such that $\sigma_{l_i}=a$. If there exists a (unique) transition $\delta(q_i,a,g)=(p,X)$ such that $\nu_i,\tau_{l_i} \models g$. Then,
\begin{itemize}
 \item $q_{i+1}=p$
\item $\nu_{i+1}(x)=\tau_{l_i}$ for all clocks $x\in X$, and $\nu_{i+1}(x)= \nu_i(x)$ otherwise.
\item $l_{i+1} = l_i+1$ if $p\in Q_L$, $l_{i+1} = l_i-1$ if $p\in Q_R$ and $l_{i+1} = l_i$  if $p\in \{t,r\}$
\end{itemize}
\item If the automaton is in a configuration $(q_i,\nu_i,l_i)$ (and $q_i\not\in \{t,r\}$) and there does not exist a transition $\delta(q_i,a,g)$ such that $\sigma_{l_i}=a$ and $\nu_i,\tau_{l_i} \models g$. Then,
\begin{itemize}
 \item $q_{i+1}=q_i$
\item $\nu_{i+1}(x)=\nu_{i}(x)$ for all clocks $x\in C$ and
\item $l_{i+1} = l_i+1$ if $q_i\in Q_L$ and $l_{i+1} = l_i-1$ if $q_i\in Q_R$
\end{itemize}
\item Termination: $q_n\in \{t,r\}$. The run is accepting if $q_n=t$ and rejecting if $q_n=r$.
\end{itemize}
Let $\mathcal F_\autm$ be a function such that $\mathcal F_\autm(\rho,\nu,i)$ gives the final configuration
$(q_n,\nu_n,l_n)$ of the unique run of $\autm$ on $\rho$ starting with the configuration $(s,\nu,i)$.
The language accepted by an automaton $\autm$ is given by 
$\mathcal L(\autm)=\{\rho ~\mid~ \mathcal F_\autm(\rho,\nu_{init},1) = (t,\nu',i), \mbox{for some } i,\nu'\}$. 
\end{definition}

The transition function satisfies the following conditions.
\begin{itemize}
\item For all $q\in Q\setminus\{t,r\}$ and $g\in G_C$, there exists $\delta(q,\rend,g)=(q',X)$ such that $q' \in Q_R\cup\{t,r\}$ and $\delta(q,\lend,g)=(q',X)$ such that $q' \in Q_L\cup\{t,r\}$ . This prevents the head from falling off the end-markers. 
\item (Determinism) For all $q \in Q$ and $a\in\Sigma'$, if there exist distinct transitions 
$\delta(q,a,g_1)=(q_1,X_1)$ and $\delta(q,a,g_2)=(q_2,X_2)$, then $g_1$ and $g_2$ are disjoint.
\end{itemize}

\begin{example}
\label{exam:potdta}
Figure \ref{exam:autex1} shows an example po2DTA. 
This automaton accepts timed words with the following property: 
There is  $b$ in the interval $[1,2]$ and a $c$ occurs before it such that, if $j$ is the position of 
the first $b$ in the 
interval $[1,2]$ then there is a $c$ exactly at the timestamp $\tau_j-1$.
\end{example}
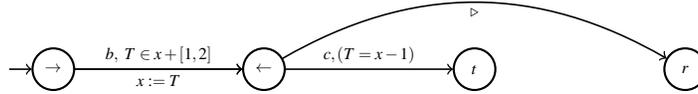
\begin{figure}
\begin{tikzpicture}[
        scale=0.7, 
        transform shape, 
        auto,
        node distance=4cm,
        semithick,
        initial text=]
      \tikzstyle{every state}=[thick]
      
\node[state,initial](S){$\rightarrow$};
\node[state](A)[right of=S]{$\leftarrow$};
\node[state](B)[right of = A]{$t$};
\node[state](C)[right of = B]{$r$};

\path[->]
(S) edge node[above] {$b,~T\in x+[1,2]$} (A)
(S) edge node[below] {$x:= T$} (A)
(A) edge node[above] {$c,(T=x-1)$} (B)
(A) edge  [bend left] node[below] {$\lend$} (C)
;
\end{tikzpicture}
\caption{Example of po2DTA}
\label{exam:autex1}
\end{figure}

\begin{definition}{[Timed Unambiguous Languages]}
The languages accepted by \potdta\/ are called Timed Unambiguous Languages (TUL).
\end{definition}

\section{From \mitlfp-fragments to \potdta}
In this section, we explore reductions from some fragments of Unary $\mathit{MITL}$ to \potdta. A powerful optimization becomes possible when dealing with the unary sublogics such as \mitlfpb\/ and \mitlfpinf. The truth of a modal formula $M_{I} \phi$ for a time point $\tau_i$ in an interval $I$ can be reduced to a simple condition involving time differences between $\tau_i$ and the times of \emph{first} and \emph{last} occurrences of $\phi$ within some  related intervals. We introduce some notation below.

\paragraph{Marking timed words with first and last $\phi$-positions\\}
Consider a formula $\phi$ in normal form, a timed word $\rho\in T\Sigma^*$ and an interval $I$. Let  $\idx^\phi_I(\rho) = \{ i \in dom(\rho) ~\mid~ \rho,i \models \phi \land \tau_i \in I\}$. Given set $S$ of positions in $\rho$ let $min(S)$ and $max(S)$ denote the smallest and largest positions in $S$, with the convention that
$min(\emptyset) =\#\rho$ and $max(\emptyset)=1$. Let $\frst^\phi_I(\rho) =
\tau_{min(\idx^\phi_I(\rho))}$ and   $\lst^\phi_I(\rho) =\tau_{max(\idx^\phi_I(\rho))}$ denote 
the times of first and last occurrence of $\phi$ within interval $I$ in word $\rho$.
If the subscript $I$ is omitted, it is assumed to be the default interval $[0,\infty)$.\\

The logic-automata translations that we give in this chapter are based on the following concepts:
\begin{itemize}
\item[i] In \cite{BMOW07}, the authors consider $\mathit{Bounded ~MTL}$ and show that the satisfiability problem for $\mathit{MITL[U_b]}$ over point-wise models is EXPSPACE-complete. This is done via translation to ATA. In \cite{BMOW08}, they show a similar result for continuous models, using model-theoretic methods, in which they construct a tableaux for the bounded formulas. The bounded size of the tableaux relies on the fact that there is a bound on the interval within which the truth of every subformula has to be evaluated. Our translation from \mitlfpb\/ also uses this concept.
\item[ii] On the other hand, \cite{MNP06} gives the translation of \mitl\/ formulas to ``Timed Transducers''. A key concept used here, is the fact that the variability within a unit interval of the truth of a subformula with non-punctual interval constraints is limited. 
\item[iii] Further, it is known that unary LTL (called Unary Temporal Logic) is expressively equivalent to \potdfa. In \ref{SShah12}, we gave a constructive reduction from \utl\/ to \potdfa. The novel concept used here is that for every \utl\/ subformula, it is sufficient to know the first and last positions in a word, where the subformula holds true. It is this concept, which justifies the expressive equivalence between the seemingly different properties of unaryness (of \utl) and determinism (of \potdfa). 
\end{itemize}
We combine the concepts (i), (ii) and (iii) described above to give translations from \mitlfpinf\/ to \potdta\/ and from \mitlfpb\/ to \potdta\/.

\subsection{From \mitlfpinf\/ to \potdta}
Fix an \mitlfpinf\/ formula $\Phi$ in normal form. We shall construct a language-equivalent \potdta\/ $\autm_\Phi$ by an inductive bottom-up construction. But first we assert an important property on which our automaton construction is based. 
\begin{lemma}\label{lem:inftychar} Given a timed word $\rho$ and $i\in dom(\rho)$,
\begin{enumerate}
 \item $\rho,i \models \fut_{[l,\infty)}\phi  \quad \fif  \quad 
  \tau_i \leq (\lst^\phi(\rho) - l) \land \tau_i <  \lst^\phi(\rho)$ 
 \item $\rho,i \models \fut_{(l,\infty)}\phi  \quad \fif  \quad 
  \tau_i < (\lst^\phi(\rho) - l)$ 
 \item $\rho,i \models \past_{[l,\infty)}\phi  \quad \fif  \quad 
  \tau_i \geq (\frst^\phi(\rho) + l) \land \tau_i >  \frst^\phi(\rho)$ 
 \item $\rho,i \models \past_{(l,\infty)}\phi  \quad \fif  \quad 
  \tau_i > (\frst^\phi(\rho) + l)$ 
\end{enumerate}
\end{lemma}
\begin{proof} We give the proof only for part (1). Remaining parts can be proved similarly.
\paragraph{Case 1} $\idx^\phi(\rho) \not= \emptyset$. Let $max(\idx^\phi(\rho))=j$.
Then, $\tau_j=\lst^\phi(\rho)$. Now, \\
$\begin{array}{ll}
       & \rho,i \models \fut{(l,\infty)} \phi \\
 \fif &  \tau_i \leq \tau_j -l \land i < j \\
\fif &  \tau_i \leq \tau_j -l \land \tau_i < \tau_j ~~~~ \mbox{(by strict monotonicity of the timed words)}\\
 \fif &  \tau_i \leq \lst^\phi(\rho) - l \land \tau_i < \lst^\phi(\rho)
\end{array}
$
\paragraph{Case 2} $\idx^\phi(\rho) = \emptyset$. We show that both LHS and RHS are false. For any $i \in dom(\rho)$ we have,  $\rho,i \not \models \fut_{[l,\infty)}\phi $. Also, $\lst^\phi(\rho)=0$. Hence,  conjunct $\tau_i < L^\phi(\rho)$ of RHS does not hold.
\end{proof}

The above lemma shows that truth of $\fut_{\langle l , \infty)} \phi$ (or $\past_{\langle l , \infty)} \phi$) at a position can be determined by knowing the value of $\lst^\phi(\rho)$ (respectively, $\frst^\phi(\rho)$). Hence for each $\fut$-type modarg $\phi$ of $\Phi$, we introduce a clock $y^\phi$ to freeze the value $\lst^{\phi}(\rho)$ and $\past$-type modarg $\phi$ of $\Phi$, we introduce a clock $x^\phi$ to freeze the value $\frst^{\phi}(\rho)$.

\begin{figure}
\begin{minipage}{0.4\linewidth}
\begin{tabular}{|l|l|}
\hline
&\\
[-1.5ex]
$\psi$ & $cond(\psi)$ \\ 
\hline
$\fut_{[l,\infty)}\phi$ & $T \leq  (y^\phi - l) ~\land~ T<y^\phi$ \\ 
\hline 
$\fut_{(l,\infty)}\phi$ & $ T <  (y^\phi - l)$ \\ 
\hline 
$\past_{[l,\infty)}\phi$ & $ T \geq  (x^\phi + l) ~\land ~T>x^\phi$ \\ 
\hline 
$\past_{(l,\infty)}\phi$ & $ T >  (x^\phi + l)$ \\ 
\hline 
\end{tabular}
\end{minipage}
\begin{minipage}{0.5\linewidth}
\begin{tikzpicture}[scale = 0.7, transform shape, auto,
        node distance=4cm,
        semithick,
        initial text=]
\tikzstyle{every state}=[thick]
      
\node[state,initial](S){$\rightarrow$};
\node[state](A)[right of=S]{$\leftarrow$};
\node[state](B)[right of=A]{$t$};

\path[->]
(S) edge node[above] {$\rend$} (A)

(A) edge[bend right] node[below] {$a_j,~\mathcal G(\phi,a_j) ~,~ y^\phi :=T$} (B)
(A) edge[bend left] node[above] {$a_1,~\mathcal G(\phi,a_1) ~,~ y^\phi :=T$} (B)
(A) edge node[above] {$\lend$} (B)
;
\end{tikzpicture}
\end{minipage}
\caption{Table for $cond(\psi)$ and automaton $\autm(\phi)$ for an $\fut$-type $\phi$.}
\label{fig:infty}
\end{figure}
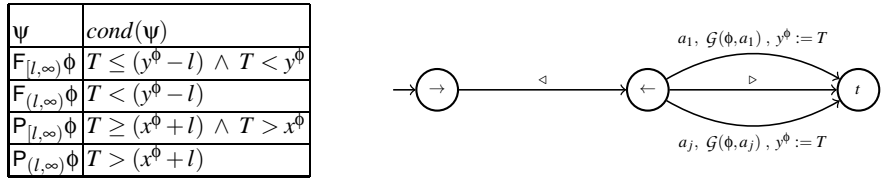

Now we give the inductive step of automaton construction. 
Consider an $\fut$-type modarg $\phi$. The automaton $\autm(\phi)$ is as shown in Figure \ref{fig:infty}. 
If $\phi = \lor_{a \in \Sigma} ~ (a~ \land {\mathscr B}_a(\{\psi_i\}))$, then for every $a\in\Sigma$, we derive the guard $\mathcal G(\phi,a)$ which is the guard on the transition labelled by $a$ in $\autm(\phi)$, such that the transition is enabled is taken if and only if $a~\land\phi$ is satisfied at that position. This is given by $\mathcal G(\phi,a) = \mathscr B_a(cond(\psi_i))$.
To define $cond(\psi_i)$, let variable $T$ denote the time stamp of current position. Then, the condition for checking truth of a modal subformula $\psi$ is a direct encoding of the conditions in lemma \ref{lem:inftychar} and is given in the table in figure \ref{fig:infty}.
 It is now straightforward to see that $\autm(\phi)$ clocks exactly the last position in the word, where $\phi$ holds.  A symmetrical construction can be given for  $\past$-type modarg $\phi$, for which $\autm(\phi)$ clocks the first position in the word where $\phi$ holds.
The following lemma states its key property which is obvious from
the construction. Hence we omit its  proof.
\begin{lemma}\label{lem:infty}
Given a  modarg $\phi$  and any timed word $\rho$, let $\nu_0$ be a valuation where $\nu_0(x^{\delta})=\frst^{\delta}(\rho)$ and 
$\nu_0(y^\delta)=\lst^{\delta}(\rho)$ for each modarg subformula $\delta$ of $\phi$, and $\nu_0(x^\phi)=\tau_{\#\rho}$ and $\nu_0(y^\phi)=0$. If
$\nu$ is the clock valuation at the end of the run of $\autm(\phi)$ starting with $\nu_0$, then
$\nu(x^\delta)=\nu_0(x^\delta)$, $\nu(y^\delta)=\nu_0(y^\delta)$ for each $\delta$, and additionally,
\begin{itemize}
\item if $\phi$ is $\fut$ modarg then $\nu(y^\phi) = \lst^\phi(\rho)$.
\item  if $\phi$ is $\past$ modarg then $\nu(x^\phi) = \frst^\phi(\rho)$.
\end{itemize}
\end{lemma}
\oomit{
\begin{proof}
Consider an $\fut$-type modarg $\phi$. The automaton $\autm(\phi)$ is as shown in figure \ref{fig:infty}. 
If $\phi = \lor_{a \in \Sigma} ~ (p_a \land {\mathscr B}_a(\{\psi_i\}))$, then for the clock $y^\phi$, and $a\in\Sigma$, we derive the guard $\mathcal G(\phi,a)$ which is the guard on the transition labelled by $a$ in $\autm(\phi)$, and which resets $y^\phi$. This is given by $\mathcal G(\phi,a) = \mathscr B_a(cond(\psi_i))$.  It is now straightforward to see that $\autm(\phi)$ clocks exactly the last position in the word, where $\phi$ holds. 
\end{proof}
}
\begin{theorem}
For any \mitlfpinf\/ formula $\Phi$, there is a language-equivalent \potdta\/ $\autm(\Phi)$ whose size is linear in the modal-DAG size of the formula. Hence, satisfiability of $\mitlfpinf$ is in \np.
\end{theorem}
\begin{proof}
Assume that $\Phi$ is in the normal form as described in Definition \ref{def:norm}. Note that reduction to normal form results in a linear blow-up in the modal-DAG size of the formula (Proposition \ref{prop:norm}). 
The construction of the complete automaton $\autm(\Phi)$ is as follows.
In an initial pass, all the $x^\phi$ clocks are
set to $\tau_{\#\rho}$. Then, the component automata $\autm(\phi)$ for clocking modargs ($\phi$) are composed in sequence with innermost modargs being evaluated first. This bottom-up construction, gives us the initial-valuation conditions at every level of induction, as required in Lemma \ref{lem:infty}. Finally, the validity of $\Phi$ at the first position may be checked. 

This construction, gives a language-equivalent \potdta\/ whose number of states is linear in the number of nodes in the DAG of $\Phi$ and the largest constant in the guards of $\autm(\Phi)$ is equal to the largest constant in the interval constraints of $\Phi$. From \cite{PS10}, we know that the non-emptiness of $\autm(\Phi)$ may be checked in \np-time. Hence we can conclude that satisfiability of \mitlfpinf\/ formulas is decidable in \np-time.
\end{proof}
 

\subsection{From \potdta\/ to \mitlfpinf}
\begin{theorem}
Given a \potdta\/ $\autm$, we may derive an equivalent \mitlfpinf\/ formula $\phi_{\autm}$ such that $\mathcal L(\autm) = \mathcal L(\phi_{\autm})$
\end{theorem}

We shall first illustrate the reduction of $\potdta$ to $\mitlfpinf$ by giving a language 
equivalent \mitlfpinf\/ formula for the \potdta\/ in Example \ref{exam:potdta}. This \potdta\/ first scans in the forward direction and clocks the first $b$ in the time interval $[1,2]$ (this is a bounded constraint), and then checks if there is a $c$ exactly 1 time unit to its past by a backward scan (this is a punctual constraint).
The automaton contains guards with both upper and lower bound constraints as well as a punctual constraints. It is critical for our reduction that the progress transitions are satisfied at unique positions in the word. 

Consider the following \mitlfpinf\/ formulas. Define $Atfirst := \neg \past\top$ as the formula which holds only at the first position in the word.\\
\hspace*{1cm}
$\phi_1 := b ~\land \past_{[1,\infty)}Atfirst ~\land \neg\past_{(2,\infty)}Atfirst$\\
\hspace*{1cm}
$\phi_2 := \phi_1 \land \neg\past_{(0,\infty)}\phi_1$\\
\hspace*{1cm}
$\Phi := \fut_{[0,\infty)} [\phi_2\land \past_{[1,\infty)}(c\land 
\neg\fut_{(1,\infty)}\phi_2 )~]$ \\
The formula $\phi_1$ holds at any $b$ within the time interval $[1,2]$. The formula $\phi_2$ holds at the \textit{unique} first $b$ in $[1,2]$. The formula $\Phi$ holds at the initial position in a word iff the first $b$ in $[1,2]$ has a $c$ exactly 1 time unit behind it. Note that the correctness of $\Phi$ relies on the \emph{uniqueness} of the position where $\phi_2$ holds. The uniqueness of the positions at which the progress transitions are taken, is the key property that allows us to express even punctual constraints (occurring in the guards of progress transitions) using only lower-bound constraints as interval constraints in the formula. It is easy to verify that the \mitlfpinf\/ formula $\Phi$ exactly accepts the timed words that are accepted by the \potdta\/ in example \ref{exam:potdta}.

\paragraph{Translation from \potdta\/ to \mitlfpinf\\}
Consider  \potdta\/ $\autm$. We shall derive a language-equivalent \mitlfpinf\/ formula $\phi_\autm$ for the automaton. Since \potdta\/ run on words that are delimited by end-markers, for the sake of simplicity in presentation, we shall derive the corresponding \mitlfpinf\/ formula over the extended alphabet $\Sigma'=\Sigma \cup \{\lend,\rend\}$. However a language equivalent formula over $\Sigma$ may be derived with minor modifications to the construction described below.

\begin{theorem}
Given a \potdta\/ $\autm$, we may derive an \mitlfpinf\/ formula $\phi_\autm$ such that $\forall \rho\in T\Sigma^* ~.~ \rho\in \mathcal L(\autm) ~\Leftrightarrow~ \rho'\in \mathcal L(\phi_\autm)$. The size of the formula is exponential in the size of the automaton.
\end{theorem}

Given any path $\pi$ of progress edges starting from the start state $s$ of $\autm$, we shall derive an \mitlfpinf\/ formula $\mathit{Enable}(\pi)$ such that the following lemma holds.

\begin{lemma}
If $\pi$ is a path of progress edges in $\autm$ which begins from the start state, we may construct an \mitlfpinf\/ formula $\mathit{Enable}(\pi)$ such that for any timed word $\rho$, there exists a partial run of $\autm$ on $\rho$ which traverses exactly the progress edges in $\pi$ and whose last transition is enabled at position $p\in dom(\rho')$ if and only if $\rho,p\models \mathit{Enable}(\pi)$.
\end{lemma}
\begin{proof}
We shall derive $\mathit{Enable}(\pi)$ by induction on the length of $\pi$. For the empty path (denoted as $<>$), we have $\mathit{Enable}(<>) ~=~ \neg\past_{(0,\infty)}\top$ which holds exactly at position 0 in $\rho'$. Now, let us inductively assume that the formula $\mathit{Enable}(\pi)$ for some path $\pi$ in $\autm$ (as shown in Figure \ref{fig:autm2form}) is appropriately constructed. We shall construct $\mathit{Enable}(\pi:e_i)$, where $\pi:e_i$ denotes the path $\pi$ that is appended with the edge $e_i$.   For each  $q$ in $\autm$, let $trans(q)$ denote the set of event-guard pairs $(a,g)$, over which a progress transition from $q$ is defined. 

Firstly, assume that each clock in $\autm$ is reset at most once\footnote{Due to the partial-ordering of the \potdta\/ and the restriction of resetting clocks only on progress edges, it is easy to see that every \potdta\/ can be reduced to one which resets every clock at most once.}. Now let $pref(\pi,x)$ denote the prefix of $\pi$ which ends with the transition that resets $x$. Hence,\\
\begin{tabular}{l|l}
&$<>$ if $x$ is not reset on any edge in $\pi$\\
$pref(\pi,x)$ =&\\
&$e_1...e_l$, which is a prefix of $\pi$ such that $x$ is reset on $e_l$.\\
\end{tabular}

\medskip
Now, given a guard $g$, we derive an \mitlfpinf\/ formula $gsat(\pi,g)$ using Table \ref{tab:guard1}. Abbreviate $\mathit{Enable}(pref(\pi,x))$ as $f(\pi,x)$\\
\begin{table}
\begin{center}
\begin{tabular}{|l|l|}
\hline
 $g$  &       $gsat(\pi,g)$\\
\hline
\hline
$0~\leq~T-x~<~c$ & $\past_{(0,\infty)}(f(\pi,x)) ~ \land~ \neg \past_{[c,\infty)}(f(\pi,x))$\\  
\hline
$0~\leq~T-x~\leq~ c$ & $\past_{(0,\infty)}(f(\pi,x)) ~ \land~ \neg \past_{(c,\infty)}(f(\pi,x))$\\
\hline
$T-x ~>~c$ & $\past_{(c,\infty)}(f(\pi,x))$\\
\hline
$T-x ~\geq~c$ & $\past_{[c,\infty)}(f(\pi,x))$\\
\hline
$T-x ~=~c$ & $\past_{[c,\infty)}(f(\pi,x)) ~\land~ \neg \past_{(c,\infty)}(f(\pi,x))$\\
\hline
$g_1~\land ~g_2$ & $gsat(\pi,g_1) ~ \land ~ gsat(\pi,g_2)$\\
\hline
$0~\leq~x-T~<~c$ & $\fut_{(0,\infty)}(f(\pi,x)) ~ \land~ \neg \fut_{[c,\infty)}(f(x))$\\
\hline  
\end{tabular}
\caption{Construction of $gsat(\pi,g)$}
\label{tab:guard1}
\end{center}
\end{table}
\begin{proposition}
\label{prop:guard}
Given any timed word $\rho$ such that there is a partial run $\pi$ of $\autm$ on $\rho$ and $\nu_\pi$ is the clock valuation of at the end of $\pi$ then $\forall p\in dom(\rho') ~.~ \rho',p\models gsat(\pi,g)$ if and only if $\nu_\pi,\tau_p\models g$. 
\end{proposition}
The proof of this proposition is directly apparent from the inductive hypothesis and the semantics of the automata. 

We may now derive $\mathit{Enable}(\pi:e_i)$ as follows. Let $e_i=(q,a_i,g_i,X_i,q_i)$.\\
\begin{itemize}
\item If $q \in Q_L$ then\\
$\mathit{Enable(\pi:e_i)} ~=~ [a_i \land gsat(\pi,g_i)]~ \land ~ [\past\mathit{Enable(\pi)}] ~\land ~$ \\
\tab \tab $[\neg\bigvee\limits_{(a,g)\in trans(q)} \past(a\land gsat(\pi,g) \land \past(\mathit{Enable}(\pi)))]$
\item If $q \in Q_R$ then\\
$\mathit{Enable(\pi:e_i)} ~=~ [a_i \land gsat(\pi,g_i)]~ \land ~ [\fut\mathit{Enable(\pi)}] ~\land ~ $\\
\tab \tab $[\neg\bigvee\limits_{(a,g)\in trans(q)} \fut(a\land gsat(\pi,g) \land \fut(\mathit{Enable}(\pi)))]$
\end{itemize}
The correctness of the above formulas may be verified by closely observing the construction. Consider the three conjuncts of the formula $\mathit{Enable}(\pi:e_i)$ in either of the above cases. The first ensures that the current position (at which the formula is evaluated) has the letter $a_i$ and satisfies the guard $g_i$ (see Proposition \ref{prop:guard}). The second conjunct ensures that the current position is to the right of (or to the left of) the position at which the partial run $\pi$ terminates (depending on whether $q \in Q_L$ or $Q_R$, respectively). The third conjunct ensures that if $p$ is the current position and $p'$ is the position at which $\pi$ terminates, then for all positions $p''$ strictly between $p$ and $p'$, none of the edges in $trans(q)$ may be enabled. Note that this is the requirement for the automaton to loop in state $q_i$ for all positions $p''$.
\end{proof}

\begin{figure}
\begin{center}
\begin{tikzpicture}
[
        scale=0.7, 
        transform shape, 
        auto,
        node distance=4cm,
        semithick,
        initial text=]
      \tikzstyle{every state}=[thick]
      
\draw (1,3) node(S) [circle,draw] {$s$};
\draw (5,3) node(Q) [circle,draw] {$q$};
\draw (7,5) node(Q1) [circle,draw] {$q_1$};
\draw (7,3) node(Qi) [circle,draw] {$q_i$};
\draw (7,1) node(Qn) [circle,draw] {$q_n$};

\draw [dotted,->] (S)-- node{$\pi$} (Q);
\draw [->] (Q)--(Q1);
\draw [->] (Q)--(Qi);
\draw [->] (Q)--(Qn);

\draw (6.4,2.8) node{$e_i$};
\draw (6,3.1) node {$a_i,g_i$};

\end{tikzpicture}
\caption{From \potdta\/ to \mitlfpinf}
\label{fig:autm2form}
\end{center}
\end{figure}

The formula $\phi_\autm$ may now be given by:\\
$\phi_\autm ~ =~ \bigvee\limits_{\wp}[\mathit{Enable}(\wp) ~ \lor~ \fut_{(0,\infty)}(\mathit{Enable}(\wp))]$\\
where $\wp$ is any path of progress edges in $\autm$ from $s$ to $t$.\\

\section{Embedding \mitlfpb\/ into \potdta}
We show a language-preserving conversion of an \mitlfpb\/ formula to a language-equivalent \potdta. 

Consider an \mitlfpb\/ formula $\Phi$ in the normal form.
We can inductively relate the truth of a subformula $\psi= \fut_{\langle l,l+1\rangle} \phi$ or
$\past_{\langle l,l+1\rangle} \phi$ within a unit interval $[r,r+1)$ to the values $\frst^\phi_I(\rho)$ 
and $\lst^\phi_I(\rho)$ of its sub-formula $\phi$ for suitable unit-length intervals 
$I$, by the following lemma\footnote{ We
shall use convention $\langle_a l,u \rangle_b$ to denote generic interval which can be open, closed or half open. Moreover, we use subscripts $a,b$ fix the openness/closeness and give generic conditions such as $\langle_a 2,3 \rangle_b + 2 = \langle_a 4,5 \rangle_b$. This instantiates to $(2,3) + 2 = (4,5)$ and $(2,3] + 2 = (4,5]$ and $[2,3) + 2 = [4,5)$ and $[2,3] + 2 = [4,5]$. Interval $[r,r)$ is empty. The same variables may also be used for inequalities. Hence, if $\langle_a$ or $\rangle_a$ is open, then $<_a$ and $>_a$ denote the strict inequalities $<$ and $>$ respectively. If $\langle_a$ or $\rangle_a$ is open, then $<_a$ and $>_a$ denote the non-strict inequalities $\leq$ and $\geq$ respectively.}. 
\begin{lemma} 
\label{lem:embedtwo}
Given a timed word $\rho$ and integers $r,l$ and $i \in dom(\rho)$ we have:
\begin{itemize}
\item $\rho,i\models \fut_{ \langle_a l, l+1 \rangle_b} \phi$ with $\tau_i \in [r,r+1)$ iff
\begin{itemize}
\item (1a) $\tau_i<\lst^\phi_{[r+l,r+l+1)}(\rho) ~\land ~ \tau_i \in [ r, (\lst^\phi_{[r+l,r+l+1)}(\rho)-l) \rangle_a$ OR 
\item (1b) $\tau_i<\lst^\phi_{[r+l+1,r+l+2)}(\rho) ~\land ~ \tau_i \in \langle_b(\frst^\phi_{[r+l+1,r+l+2)}(\rho)-(l+1)),(r+1))$ \\
\end{itemize}
\item $\rho,i\models \past_{ \langle_a l, l+1 \rangle_b} \phi$ with $\tau_i \in [r,r+1)$ iff
\begin{itemize} 
\item (2a) $\tau_i>\frst^\phi_{[r-l-1,r-l)}(\rho) ~\land ~ \tau_i \in ~ [ r,(\lst^\phi_{[r-l-1,r-l)}(\rho)+l+1) \rangle_b$ OR
\item (2b) $\tau_i>\frst^\phi_{[r-l,r-l+1)}(\rho) ~\land~ \tau_i\in ~ \langle_a (\mathcal F^\phi_{[r-l,r-l+1)}(\rho)+l),(r+1) )$ 
\end{itemize}
\end{itemize}
\end{lemma}
\begin{proof}
This lemma may be verified using the figure \ref{fig:bmitlpo2fig1}. We consider the case of $\fut_{\langle_a l, l+1 \rangle_b} \phi$ omitting the symmetric case of $\past_{\langle_a l, l+1 \rangle_b} \phi$. Let $\psi=\fut_{ \langle_a l, l+1 \rangle_b} \phi$.
Fix a timed word $\rho$.\\
\textbf{Case 1}:  (1a) holds. (We must show that $\rho,i \models \psi$ and $\tau_i\in [r,r+1)$). Since conjunct 1 holds, clearly $\idx^\phi_{[r+l,r+1+1)} \not=\emptyset$ and it has max element $j$ such that $\tau_j=\lst^\phi_{[r+l,r+l+1)}(\rho)$ and $\rho,j \models \phi$ and $i<j$. Also, by second conjunct of (1a) $\tau_i \in [r,\tau_j -l \rangle_a$. Hence, by examination of Figure \ref{fig:bmitlpo2fig1}, $\tau_i \in \tau_j - \langle_a l , l+1 \rangle_b$  and hence $\rho,i \models \psi$. \\
\textbf{Case 2}:  (1b) holds. (We must show that $\rho,i \models \psi$ and $\tau_i\in [r,r+1)$). Since conjunct 1 holds, clearly $\idx^\phi_{[r+l+1,r+1+2)} \not=\emptyset$ and it has min element $j$ such that $\tau_j=\frst^\phi_{[r+l+1,r+l+2)}(\rho)$ and $\rho,j \models \phi$. From Figure \ref{fig:bmitlpo2fig1}, $j$ is a witness such that for all $k$ such that $(\tau_j-(l+1)) ~<_b~\tau_k~<_a~(\tau_j-l)$, we have $\tau_k\models\phi$. Therefore, by second conjunct of (1b) $\tau_i \in \langle_b \tau_j -(l+1), (r+1))$, we may infer that $\tau_i \in \tau_j - \langle_a l , l+1 \rangle_b$. Hence $\rho,i \models \psi$ and $\tau_i\in [r,r+1)$.\\
\textbf{Case 3}: $\rho,i \models \psi$ and $\tau_i\in [r,r+1)$ and  first conjunct of (1b) does not hold. (We must show that (1a) holds.) Since $\tau_i \not < \lst^\phi_{[r+l+1,r+l+2)}(\rho)$ we have $\idx^\phi_{[r+l+1,r+l+2)}(\rho)=\emptyset$. \\
Since $\rho,i \models \psi$ for some $\tau_i \in [r,r+1)$, there is some $\tau_j>\tau_i$ s.t. $\rho,j \models \phi$ and $r+l \leq \tau_j < r+l+1$ as well as $\tau_j \in \tau_i + \langle_a l,l+1 \rangle_b$, and 
$\tau_j \leq \lst^\phi_{[l+r,l+r+1)}(\rho)$. Hence, we have $\tau_i \in \lst^\phi_{[l+r,l+r+1)}(\rho) - \langle_a l, l+1 \rangle_b$ which from Figure
\ref{fig:bmitlpo2fig1} gives us that $\tau_i \in [r,\lst^\phi_{[l+r,l+r+1)}(\rho) -l \rangle_a$. Also, we can see that $\tau_i < \lst^\phi_{[l+r,l+r+1)}(\rho)$. Hence (1a) holds.\\
\textbf{Case 4}: $\rho,i \models \psi$ and $\tau_i\in [r,r+1)$ and first conjunct of (1b) holds but the second conjunct of (1b) does not hold. (We must show that (1a) holds.)
As $\rho,i \models \psi$ for some $\tau_i \in [r,r+1)$, there is some $\tau_j>\tau_i$ s.t. $\tau_j \in \tau_i + \langle_a l,l+1 \rangle_b$ and $\rho,j \models \phi$. \\
Since $\tau_i <  \lst^\phi_{[r+l+1,r+l+2)}(\rho)$ we have
$\idx^\phi_{[r+l+1,r+l+2)}(\rho) \not= \emptyset$. However, second conjunct of
(1b) does not hold. Hence, $\tau_i \not >_b \frst^\phi_{[r+l+1,r+l+2)}(\rho) - (l+1)$ and $j\in [r+l,r+l+1)$.
By examination of Figure \ref{fig:bmitlpo2fig1}, we conclude that 
$\idx^\phi_{[r+l,r+1+1)} \not=\emptyset$ and $\tau_j \leq \lst^\phi_{[r+l,r+l+1)}(\rho)$. Hence, $\tau_j-l \leq \lst^\phi_{[r+l,r+l+1)}(\rho)-l$.
 This gives us that
$\tau_i \in [r,\tau_j -l\rangle_a$ (see Figure \ref{fig:bmitlpo2fig1}).
Thus, (1a) holds.
\end{proof}

\begin{figure}
\begin{tikzpicture}[scale=0.74, transform shape]
\draw[dotted](-1,5) -- (0,5);
\draw(0,5) node{I} -- (4,5) node{I}; \draw[dotted](4,5)-- (7,5) node{I}; \draw (7,5) -- (11,5) node{I} -- (15,5) node{I};
\draw(0,4.7) node{$r$}; \draw(4,4.7) node{$r+1$}; \draw(7,4.7) node{$r+l$}; \draw(11,4.7) node{$r+l+1$}; \draw(15,4.7) node{$r+l+2$};
\draw[dotted](1.5,5) --(1.5,6) node{$(y-l)$}; \draw[dotted](2.5,5) --(2.5,5.5) node{$x-(l+1)$}; \draw[dotted](8.5,5) --(8.5,6) node{$y$}; \draw[dotted](13.5,5) --(13.5,5.5) node{$x$};
\draw[dotted](1.5,5) --(1.5,4); \draw[dotted](2.5,5) --(2.5,4); \draw[dotted](-1,5) --(-1,4);\draw[dotted](5.5,5) --(5.5,4); \draw[dotted](8.5,5) --(8.5,4); \draw (8.4,3.7) node{$\mathcal L^\phi_{r+l}$}; \draw[dotted](13.5,5) --(13.5,4); \draw (13.5,3.7) node{$\mathcal F^\phi_{r+l+1}$};
\draw(0.2,4) node{$\psi$}; \draw(4.3,4) node{$\psi$}; 
\draw[<->] (1.5,4) -- (2.5,4);\draw (-1,4.3) -- (1.5,4.3); \draw (2.5,4.3) -- (5.5,4.3); \draw[(-)] (8.5,4.2) -- (13.5,4.2);
\draw (-0.9,4.3) node{$\langle_b$};
\draw (1.6,4.3) node{$\rangle_a$};
\draw (2.6,4.3) node{$\langle_b$};
\draw (5.6,4.3) node{$\rangle_a$};
\draw (2,3.7) node{$\neg \psi$};
\draw (11,4) node{$\neg \phi$};
\end{tikzpicture}
\caption{Case of $\psi := \fut_{\langle_a l,l+1 \rangle_b}\phi$}
\label{fig:bmitlpo2fig1}
\end{figure}

\begin{figure}
\begin{tikzpicture}[scale=0.8, transform shape]
\draw(0,5) node{I} -- (4,5) node{I} -- (8,5) node{I};  \draw[dotted](8,5)-- (11,5)node{I};  \draw (11,5) -- (15,5) node{I}; \draw[dotted](15,5)-- (16,5) node{};
\draw(0,4.7) node{$r-l-1$}; \draw(4,4.7) node{$r-l$}; \draw(8,4.7) node{$r-l+1$}; \draw(11,4.7) node{$r$}; \draw(15,4.7) node{$r+1$};
\draw[dotted](12.5,5) --(12.5,6) node{$(y+l+1)$}; \draw[dotted](13.5,5) --(13.5,5.5) node{$x+l$}; \draw[dotted](1.5,5) --(1.5,6) node{$y$}; \draw[dotted](6.5,5) --(6.5,5.5) node{$x$};
\draw[dotted](12.5,5) --(12.5,4); \draw[dotted](13.5,5) --(13.5,4); \draw[dotted](9.5,5) --(9.5,4);\draw[dotted](16,5) --(16,4); \draw[dotted](1.5,5) --(1.5,4) node{$\mathcal L^\phi_{r-l-1}$}; \draw[dotted](6.5,5) --(6.5,4) node{$\mathcal F^\phi_{r-l}$};
\draw(11.2,4) node{$\psi$}; \draw(15.3,4) node{$\psi$}; 
\draw[<->] (12.5,4) -- (13.5,4);\draw (9.5,4.3) -- (12.5,4.3); \draw (13.5,4.3) -- (16,4.3); \draw[<->] (1.7,4.2) -- (6.3,4.2);
\draw (9.6,4.3) node{$\langle_a$};
\draw (12.6,4.3) node{$\rangle_b$};
\draw (13.6,4.3) node{$\langle_a$};
\draw (16.1,4.3) node{$\rangle_b$};
\draw (13,3.7) node{$\neg \psi$};
\draw (4,4) node{$\neg \phi$};
\end{tikzpicture}
\caption{Case of $\psi :=\past_{\langle_a l,l+1 \rangle_b}\phi$}
\label{fig:bmitlpo2fig2}
\end{figure}

From Lemma \ref{lem:embedtwo}, we can see that in order to determine the truth of a formula of the form $\psi = \fut_{\langle l,l+1 \rangle} \phi$ at any time stamp in $[r,r+1)$, it is sufficient to clock the first and last occurrences of $\phi$ in the intervals $[r+l, r+l+1)$ and $[r+l+1,r+l+2)$. Similarly, in order to determine the truth of a formula of the form $\psi = \past_{\langle l,l+1 \rangle} \phi$ at any time stamp in $[r,r+1)$, it is sufficient to clock the first and last occurrences of $\phi$ in the intervals $[r-l, r-l+1)$ and $[r-l-1,r-l)$.

\usetikzlibrary{automata}%
\usetikzlibrary{positioning}
\begin{figure}
\begin{tikzpicture}[
        scale=0.9, 
        transform shape, 
        auto,
        node distance=3.5cm,
        semithick,
        initial text=]
      \tikzstyle{every state}=[thick]
      
\node[state,initial](S){$\leftarrow$};
\node[state](A)[right of=S]{$\rightarrow$};
\node[state](B2)[right of=A]{$\rightarrow$};
\node[state](C2)[right of=B2]{$\leftarrow$};
\node[state](T)[right of=C2]{$t$};

\path[->]
(S) edge node[above] {$\lend$} (A)
(A) edge[out=100, in=100, looseness=0.5] node {$\rend$} (T)

(A) edge[bend right] node[above] {$x^\phi_{[r,r+1)} :=T$} (B2)
(A) edge[bend right] node[swap] {$a_j,~\mathcal G(\phi,[r,r+1),a_j)$} (B2)

(B2) edge node[above] {$\rend$} (C2)

(B2) edge node[above] {$\rend$} (C2)
(C2) edge[bend right] node[above]  {$y^\phi_{[r,r+1)}:=T$} (T)
(C2) edge[bend right] node[swap]{$a_j, ~\mathcal G(\phi,[r,r+1),a_j)$} (T)
;
\end{tikzpicture}
\caption{ Automaton $\autm(\phi, [r,r+1))$}
\label{bmitlfp:autm}
\end{figure}
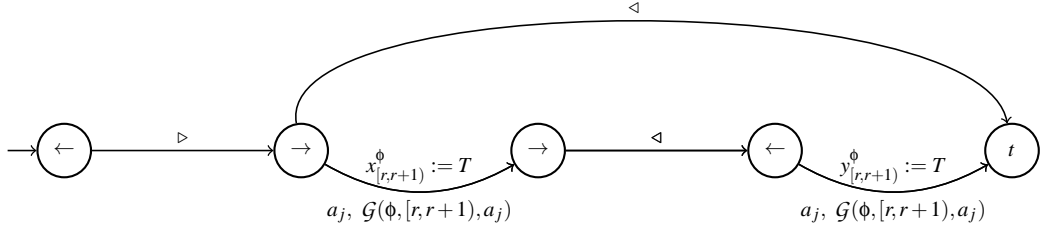

The automaton $\autm(\Phi)$ is constructed in an inductive, bottom-up manner as follows.
For every modarg $\phi$ of $\Phi$, we first inductively evaluate the set of unit intervals within which its truth must be evaluated. Each such requirement is denoted by a tuple $(\phi,[r,r+1))$. This is formalized as a closure set of a subformula with respect to an interval.  For an interval $I$, let $spl(I)$ denote a partition set of $I$, into unit length intervals. For example, if $I= (3,6]$ then $spl(I) = \{(3,4),[4,5),[5,6]\}$. The closure set may be built using the following rules.
\begin{itemize}
 \item $Cl(\phi,[r,r+1)) = \{(\phi,[r,r+1))\} \cup_j Cl(\psi_j,[r,r+1))$,\\
where $\{\psi_j\}$ is the set of immediate modal subformulas of $\phi$.
 \item $Cl(\fut_I \phi, [r,r+1)) = \cup_{\langle l,l+1\rangle\in spl(I)}Cl(\fut_{\langle l,l+1\rangle},[r,r+1))$
\item $Cl(\past_I \phi, [r,r+1)) = \cup_{\langle l,l+1\rangle\in spl(I)}Cl(\past_{\langle l,l+1\rangle},[r,r+1))$
 \item $Cl(\fut_{\langle l,l+1\rangle} \phi, [r,r+1)) = Cl(\phi,[r+l,r+l+1)) \cup Cl(\phi,[r+l+1,r+l+2))$
 \item $Cl(\past_{\langle l,l+1\rangle} \phi, [r,r+1)) = Cl(\phi,[r-l-1,r-l)) \cup Cl(\phi,[r-l,r-l+1))$
\end{itemize}
Define strict closure $SCl(\phi,[r,r+1)) = Cl(\phi,[r,r+1)) \setminus \{(\phi,[r,r+1))\}$. The following lemma states the key property of $\autm(\phi, [r,r+1))$.

\begin{table}
\begin{center}
\begin{tabular}{|l|l|}
\hline
&\\
[-1.5ex]
$\psi$ & $cond(\psi,[r,r+1))$ \\ 
[1ex]
\hline
&\\
[-1.5ex]
$\fut_{\langle_a l,l+1 \rangle_b} \delta$ & $T<y^\delta_{[r+l,r+l+1)} ~\land ~T \in [ r,(y^\delta_{[r+l,r+l+1)}-l)\rangle_b ~~\lor$ \\
                       &$T<y^\delta_{[r+l+1,r+l+2)} ~\land ~ T \in 
                       \langle_a (x^\delta_{[r+l+1,r+l+2)}-(l+1)),(r+1))$ \\ 
[1ex]
\hline 
&\\
[-1.5ex]
$\past_{\langle_a l,l+1 \rangle_b} \delta$ & $T>x^\delta_{[r-l-1,r-l)} ~\land ~ T \in ~ [ r,(y^\delta_{[r-l-1,r-l)}+l+1) \rangle_a ~~\lor$ \\
                       & $T>x^\delta_{[r-l,r-l+1)} ~\land~ T\in ~ 
                        \langle_b (x^\delta_{[r-l,r-l+1)}+l), (r+1))$ \\ 
[1ex]
\hline 
&\\
[-1.5ex]
$\fut_I(\delta)$ & $ \bigvee_{\langle l,l+1\rangle\in spl(I)}(cond(\fut_{\langle l,l+1\rangle}\delta,[r,r+1)))$ \\
[1ex]
\hline 
&\\
[-1.5ex]
$\past_I(\delta)$ & $ \bigvee_{\langle l,l+1\rangle\in spl(I)}(cond(\past_{\langle l,l+1\rangle}\delta,[r,r+1)))$ \\
[1ex]
\hline 
\end{tabular}
\caption{Construction of $cond(\psi,[r,r+1))$}
\label{tab:guard}
\end{center}
\end{table}

\begin{lemma}
\label{lem:embedtwob}
For any modarg $\phi$ in normal form, timed word $\rho$ and integer $r$, we construct an automaton $\autm(\phi,[r,r+1))$ such that, if the initial clock valuation is $\nu_0$, with 
$\nu_0(x^\delta_I)=\frst^\delta_I(\rho)$ and $\nu_0(y^\delta_I)=\lst^\delta_I(\rho)$ for all $(\delta,I)\in SCl(\phi,[r,r+1))$ and $\nu_0(x^\phi_{[r,r+1)})=\#\rho$ and $\nu_0(y^\phi_{[r,r+1)})=0$, then 
the automaton $\autm(\phi,[r,r+1))$ will accept $\rho$ and terminate with a valuation $\nu$ such that
\begin{itemize}
 \item $\nu(x^{\phi}_{[r,r+1)}) =\frst^{\phi}_{[r,r+1)}(\rho)$
 \item $\nu(y^{\phi}_{[r,r+1)}) =\lst^{\phi}_{[r,r+1)}(\rho)$  and
\item $\nu(c) = \nu_0(c)$, for all other clocks.
\end{itemize}
\end{lemma}

\begin{proof}[Proof sketch]
The automaton $\autm(\phi,[r,r+1))$ is given in Figure \ref{bmitlfp:autm}. For each $(\delta,I) \in SCl(\phi,[r,r+1))$, $\autm(\phi,[r,r+1)$ uses the  clock values of $x^{\delta}_{I}$ and $y^{\delta}_{I}$ in its guards, and it resets the clocks $x^\phi_{[r,r+1)}$ and $y^\phi_{[r,r+1)}$. 

For every $\psi$, which is an immediate modal subformula of $\phi$, we derive $cond(\psi,[r,r+1))$ as given in Table \ref{tab:guard}. The first two rows in Table \ref{tab:guard} are directly adapted from Lemma \ref{lem:embedtwo}. The last two rows, may be easily inferred from the semantics of \mitlfpb. Hence, we may infer that $\forall i\in dom(\rho)$ if $\tau_i\in [r,r+1)$ then $\nu_0,\tau_i\models cond(\psi,r)$ iff $\rho,i\models\psi$. Now, if $\phi = \bigvee\limits_{a\in\Sigma} (a \land \mathscr B_a(\psi_i))$, then the guard on the transitions labelled by $a$, which reset $x^\phi_{[r,r+1)}$ and $y^\phi_{[r,r+1)}$ (as in figure \ref{bmitlfp:autm}) is given by $\mathcal G(\phi,[r,r+1),a) = \mathscr B_a(cond(\psi_i,[r,r+1)))$. It is straightforward to see that $\forall i\in dom(\rho)$ if $\tau_i\in [r,r+1)$ and $\sigma_i=a$ then $\nu_0,\tau_i\models\mathcal G(\phi,[r,r+1),a)$ iff $\rho,i\models\phi$. By observing the \potdta\/ in figure  \ref{bmitlfp:autm},  we can infer that it clocks the first and last $\phi$-positions in the unit interval $[r,r+1)$, and respectively assigns it to $x^\phi_{[r,r+1)}$ and $y^\phi_{[r,r+1)}$. 
\end{proof}

\begin{theorem}
Given any \mitlfpb\/ formula $\Phi$, we may construct a \potdta\/ which is language-equivalent to $\Phi$. Satisfiability of \mitlfpb\/ formulas is decidable in \nexp-time.
\end{theorem}
\begin{proof}
Firstly, $\Phi$ is reduced to the normal form, as described in section \ref{sec:mitlsem}. The automaton is given by $\autm_\Phi = \autm_{reset};\autm_{induct};\autm_{check}$ \footnote{The operator ``;'' denotes sequential composition of \potdta}.  The \potdta\/ $\autm_{reset}$ makes a pass to the end of the word and resets all $x^\delta_I$ (for all $(\delta,I)\in Cl(\Phi,[0,1))$) to the value $\tau_{\#\rho}$. The bottom-up arrangement of $Cl(\Phi,[0,1))$ is the sequence of elements $(\delta,I)$ of the set in bottom-up order of the subformulas $\delta$ (in any arbitrary ordering of the intervals). $\autm_{induct}$ sequentially composes $\autm(\delta,I)$ in the bottom-up sequence of $Cl(\Phi,[0,1))$. This bottom-up sequence ensures that the conditions for the initial valuation of each of the component automata, as required in lemma \ref{lem:embedtwob}, are satisfied.
Finally, $\autm_{check}$ checks if the clock value $x^\Phi_{[0,1)}=0$, thereby checking the validity of $\Phi$ at the first position in the word. 

\noindent
\emph{Complexity}: Assuming DAG representation of the formula, reduction to normal form only gives a linear blow up in size of the DAG. Observe that the $Cl(\psi,[r,r+1))$  for $\psi=\fut_I (\phi)$ or  $\past_I(\phi)$ contains $m+1$ number of elements of the form $\phi,[k,k+1)$, where $m$ is the length of the interval $I$. Hence, if interval constraints are encoded in binary, it is easy to see that the size of $Cl(\Phi,[0,1))$ is $O(2^l)$, where $l$ is the modal DAG-size of $\Phi$. 
Since each  $\autm(\phi,[r,r+1))$ has a constant number of states, we may infer that the number of states in $\autm(\Phi)$ is $O(2^l)$. 
Since the non-emptiness of a \potdfa\/ may be decided with NP-complete complexity, we conclude that satisfiability of a \mitlfpb\/ formula is decidable with NEXPTIME complexity. 
\end{proof}

\section{Decision Complexities for \mitlfp\/ Fragments}
Figure \ref{fig:complexity} depicts the satisfaction complexities of various unary sublogics of \mitl\/ that are studied in this chapter.
We shall use tiling problems (\cite{vEB97},\cite{Fur83}) to derive lower bounds for satisfiability problems of the logics considered.
A tiling system $(X,M_H,M_V)$  consists of a finite set of tile types $X$ 
with $M_V,M_H \subseteq X \times X$.  Tiling of a rectangular region of size 
$p \times q$ is a map $T:\{1, \ldots, p\} \times \{1, \ldots, q\} \rightarrow X$ 
such that $(T(i,j),T(i+1,j)) \in M_H$ and $(T(i,j),T(i,j+1)) \in M_V$. These are called 
horizontal and vertical matching constraints. An instance of the tiling problem 
specifies the region to be tiled  with a given tiling system and additional 
constraints on tiling, if any (such as $T(1,1) = a \land T(p,q)=b$).

We reduce tiling problems to satisfiability of $\mitlfp$ formulae. Thus, a tiling $T$
is represented by a timed word $\rho_T$ over the alphabet $X \cup {s}$ such that the sequence of letters is just catenation of rows of $T$ separated by a fresh separator letter $s$. Hence, length of $\rho_T= p \times (q+1)$ and $s$ occurs at positions $i(p+1)$ with $1 \leq i \leq q$. Depending upon the logic in consideration, various schemes are selected for time stamping the 
letters of  $\rho_T$ so that horizontal and vertical matching constraints can be enforced. We shall use abbreviations $XX \df \lor_{a \in X} (a)$ and $XXS \df XX \lor s$ and $Atlast \df \neg\fut_{[0,\infty)}  XXS$ in the formulae.

\paragraph{\expspace-hard tiling problem} Given a problem instance consisting of a tiling system
$(X,M_H,M_V)$, a natural number $n$ and first and final tiles $f$ and $t$, the solution of the problem is a tiling $T$ of a rectangle of size $2^n\times m$ such that $T(1,1)=f$ and $T(2^n,m)=t$, for some natural number $m>0$. This tiling problem is known to be \expspace-hard in $n$.\\

\begin{theorem}\label{theo:expspacehard} 
Satisfiability of $\mitlf$ (and hence $\mitlfp$) is \expspace-hard.
\end{theorem}
\begin{proof}
We represent a tiling $T$ by a timed word $\rho_T$ where the time stamps of the letters are exactly $0,1,2, \ldots, 2^n\times(m+1)-1$. Consider the $\mitlf$ formula $\Phi_{EXPSPACE}$ as conjunction $\phi_1 \land \phi_f \land \phi_t \land \phi_s \land \phi_H\land \phi_V$ of formulae given below. 

\medskip

\noindent $\phi_1 := G [XXS \Rightarrow~ ((\neg\fut_{(0,1)}XXS) ~ \land \fut_{(0,1]} XXS ~\lor Atlast))]$\\
$\phi_s := \fut_{(2^n-1,2^n]} s ~\land ~ G [s ~\Rightarrow \{\neg(\fut_{(0,2^n]}s) \land (\fut_{(2^n,2^n+1]}s\lor ~Atlast)\}]$\\
$\phi_f := f$\\
$\phi_t := \fut [t\land\fut_{(0,1]}(s\land~ Atlast)]$\\ 
$\phi_H := G [\bigwedge\limits_{a\in X}\{a\Rightarrow~ \fut_{(0,1]}(s \lor \bigvee\limits_{(a,b) \in M_H}~b]$\\
$\phi_V := G [\bigwedge\limits_{a\in X}\{a\Rightarrow~ (( \fut_{(0,2^n+1]} Atlast) ~\lor~ (\fut_{(2^n,2^n+1]}(\bigvee\limits_{(a,b)\in M_V} ~b))]$
\medskip

Conjunct $\phi_1$ ensures that letters occur exactly at integer time points. Formula $\phi_s$ indicates that the first separator $s$ occurs at time-point $2^n$ and subsequently $s$ repeats exactly after a time distance of $2^n+1$. $\phi_H$ and $\phi_V$ respectively encode horizontal and vertical matching rules.
 Note that a letter and its vertically above letter occur at time distance $2^n+1$ and this is used for enforcing vertical compatibility.
It is clear from the formula construction that  $\Phi_{EXPSPACE}$
is satisfiable iff the original tiling problem has a solution. The size of $\Phi_{EXPSPACE}$ is linear in $n$ since we use binary encoding of time constants. Hence, we conclude that satisfiability of $\mitlf$ is \expspace-hard.
\end{proof}

\paragraph{\nexptime-hard tiling problem} Given a problem instance consisting of a tiling system
$(X,M_H,M_V)$, a natural number $n$ and a sequence $t=t_1,\ldots,t_n$ of leftmost $n$ tiles in bottom row, a solution to the problem is a tiling $T$ of a square of size $2^n\times 2^n$ such that $T(1,j)=t_j$ for $1 \leq j \leq n$. This tiling problem is known to be \nexptime-hard in $n$.

\begin{theorem} 
\label{theo:nexptimehard} Satisfiability of $\mitlfb$ (hence
$\mitlfpb$)  is \nexptime-hard.
\end{theorem}
\begin{proof}
The encoding of a tiling in timed word is exactly same as in Theorem \ref{theo:expspacehard}. Thus, letters occur at successive integer times and
the first $l=2^n\times (2^n+1)$ letters encode
the tiling. Remaining letters (if any) are arbitrary and ignored. 
The timestamp of $s$ ending the last row of tiling is $l-1$. All the letters denoting the last row occur in the closed interval $I_{last} = [l-1-(2^n+1),l-1]$ and letters denoting non-last row occur in the half open time interval $I_{nonlast} =[0,l-1-(2^n+1))$. 

The formula $\Phi_{NEXPTIME}$ is  similar to formula $\Phi_{EXPSPACE}$ but
all unbounded modalities $F \psi$ and $G\psi$ are replaced by bounded modalities $F_{[0,l-1]} \psi$ and $G_{[0,l-1]}$ and $Atlast$ is omitted. Instead we use time interval $I_{nonlast}$ so that $G_{nonlast}=G_{I_{nonlast}}$ and so on.
\medskip

\noindent
$\phi_1 := G_{[0,l-2]} [XXS \Rightarrow~ (\neg\fut_{[0,1)}(XXS) ~ \land \fut_{[0,1]}(XXS))]$\\
$\phi_s := \fut_{[0,2^n)} \neg s ~\land \fut_{[0,2^n]} s ~\land~ 
    G_{nonlast} [ s ~\Rightarrow \{\neg(\fut_{[0,2^n]}s) \land      (\fut_{[0,2^n+1]}s)\}]$\\
$\phi_t := t_1 \land ~ \fut_{[0,1]} (t_2 \land \fut_{[0,1]}( ...\fut_{[0,1]}(t_n)))$\\
$\phi_H := G_{[0,l-2]}[\bigwedge\limits_{a\in X}(a\Rightarrow~ \fut_{[0,1]} ~(s \lor \bigvee\limits_{(a,b) \in M_H} b))]$\\
$\phi_V := G_{nonlast}[\bigwedge\limits_{a\in X}(a\Rightarrow~ \fut_{(2^n,2^n+1]}(\bigvee\limits_{(a,b) \in M_V} b))]$
\medskip

Conjunct $\phi_1$ (together with $\phi_t$) ensures that letters in interval $[0,l-1]$ occur only at integer time points. $\phi_t$ ensures that the first $n$ tiles match $t$. Remaining conjuncts are similar to those in Theorem \ref{theo:expspacehard}.

It is easy to see that $\Phi_{NEXPTIME}$ is satisfiable iff the original tiling problem 
has a solution. The size of $\Phi_{NEXPTIME}$ is linear in $n$ since constant $l$ 
can be coded in binary in size linear in $n$. Hence, we  conclude that satisfiability 
of $\mitlfb$ is \nexptime-hard. 
\end{proof}

\paragraph{\pspace-hard tiling problem} (Corridor Tiling) A problem instance of the Corridor Tiling problem consists of a tiling system $(X,M_H,M_V)$ and a natural number $n$, subsets $W_l, W_r \subseteq X$ of tiles which can occur on left and right boundaries of the tiling region, and sequences \textit{top}=$t_1t_2 \ldots t_n$ and \textit{bottom}=$b_1b_2 \ldots b_n$ of tiles of length $n$ each. A solution to this problem is a tiling $T$ of a rectangle of size $n\times m$, for some natural number $m>0$, such that the bottom row is \textit{bottom}, and the top row is \textit{top}. Moreover only tiles from $W_l$ and $W_r$ can occur at the start and end of a row respectively. This problem is known to be \pspace-hard in $n$.

\begin{theorem} \label{theo:pspacehard} 
Satisfiability of $\mitlfz$ (and hence also $\mitlfzinf$ and $\mitlfpzinf$) is \pspace-hard.
\end{theorem}
\begin{proof}
We represent a tiling $T$ by a timed word $\rho_T$ where the first letter is at time $0$ and time distance between successive letters is within the open interval $(1,2)$. Consider the $\mitlfz$ formula $\Phi_{PSPACE}$ as conjunction of formulae given below. Note that over strictly monotonic time words $F_{[o,u\rangle} \phi$ is equivalent to $F_{(o,u\rangle} \phi$.
 \begin{itemize}
 \item $\phi_1 := G[XXS \implies ((\neg\fut_{[0,1]}XXS) \land (Atlast \lor \fut_{[0,2)}XXS)]$ ensures that successive events occur at time distance $(1,2)$.
 \item A row is of length $n$ \\
$\phi_n := G[s\implies \{Atlast ~ \lor 
(\fut_{[0,2)} (XX \land \fut_{[0,2)} (XX \land  (\ldots n times \ldots \land \fut_{[0,2)}(s))\dots))]$
\item
 Horizontal Compatibility: $\phi_H := G[\bigwedge\limits_{a\in X} (a \implies
  \fut_{[0,2)}\{\bigvee\limits_{(a,b) \in M_H} b ~\lor s\})]$
 \item
 Vertical Compatibility:  Formula $a \land \fut (s \land \fut s)$  denotes a tile $a$ in  row other than the last row. Hence \\
  $\phi_V := \bigwedge\limits_{a\in X}~~G [
 \{a\land \fut(s\land\fut s)\}   \implies (\fut_{[0,2)} (XXS \land \fut_{[0,2)} XXS \land ( \ldots n times... \land\fut_{[0,2)}(\bigvee\limits_{(a,b)\in M_V} b)] $
 \item Matching the bottom row:\\
   $\phi_B := t_1 \land (\fut_{[0,2)} t_2 \land \fut_{[0,2)} t_3 \land \ldots  \fut_{[0,2)} (t_n \land \fut_{[0,2)} s)) \ldots ))$.
 \item Matching the top segment:\\
   $\phi_T :=\fut[ s \land \fut_{[0,2)} (
    b_1 \land (\fut_{[0,2)} b_2 \land \fut_{[0,2)} b_3 \land \ldots  
    \fut_{[0,2)} (b_n \land \fut_{[0,2)} s \land Atlast)) \ldots )]$.
 \item Matching white on the left side of the tiling:\\
  $\phi_L :=G[s\implies (Atlast \lor \fut_{[0,2)}(\bigvee\limits_{a\in W_l} a))]$
 \item Matching white on the right side of the tiling:\\
 $\phi_R:= \neg  F[\{\bigvee\limits_{a\not\in  W_r} a\} \land \fut_{[0,2)} s\}]$
\end{itemize}  
It is clear from the formula construction that  $\Phi_{PSPACE}$
is satisfiable iff the original tiling problem has a solution. The size of $\Phi_{PSPACE}$ is linear in $n$. Hence, satisfiability of $\mitlfz$ is 
\pspace-hard. 
\end{proof}

\section{Expressiveness of \mitlfp\/ Fragments}
\label{app:express}
The relative expressiveness of the fragments of Unary \mitlfp\/ is as depicted in Figure \ref{fig:unaryexpress}. The figure also indicates the languages considered to separate the logics expressively.

\begin{theorem}
\label{thm:express1}
\mitlfpb\/ $\subsetneq$ \mitlfpinf.
\end{theorem}
\begin{proof}
(i) \mitlfpb\/ $\subseteq$ \mitlfpinf.\\
This is evident from the translation of \mitlfpb\/ formulas to equivalent \potdfa, and the equivalence between \potdfa, \ttl\/ and \mitlfpinf.\\
(ii) \mitlfpinf $\nsubseteq$ \mitlfpb\\
Consider the \mitlfpinf\/ formula $\phi = \fut_{(0,\infty]}(a \land \fut_{(2,\infty)}c)$. We can show that there is no equivalent \mitlfpb\/ formula to $\phi$.
We shall prove this using a \mitlfpb\/ EF game \cite{PS11}, with $n$ rounds and $\maxint = k$, for any $k>0$. In such a game, the \ssp\/ is allowed to choose only non-singular intervals of the form $\langle l,u\rangle$, such that $u\leq k$. Let $m=nk$. $\aaa_{n,k}$ and $\bbb_{n,k}$ two families of words such that $untime(\aaa_{n,k})=untime(\bbb_{n,k})= a^{m+1}c$ and each $a$ occurs at integer timestamps from $0$ to $m$, and $c$ occurs beyond $m+2$ in $\aaa_{n,k}$ and before $m+2$ in $\bbb_{n,k}$. Then, $\forall n,k ~.~ \aaa_{n,k} \models \phi$ and $\bbb_{n,k}\not\models\phi$. Since, the \ssp\/ will be unable to place its pebble on the last $c$, the \ddp\/ has a copy-cat winning strategy for an $n$-round \mitlfpb\/ EF game with $\maxint = k$, over the two words. 
\end{proof}

\begin{theorem}
\mitlfpzinf $\subsetneq$ \mitlfp
\end{theorem}
\begin{proof}
Consider the \mitlfp\/ formula $\phi:= \fut_{(0,\infty)}[a\land \fut_{(1,2)}c]$. Note that this formula requires any $a$ in the word to be followed by a $c$ within $(1,2)$ time units from it. While either one of the bounds of the interval $(1,2)$ (either lower bound of 1 t.u. or upper bound of 2 t.u.) may be specified by a \mitlfpzinf\/ modality (such as $\fut_{(1,\infty)}$ or $\fut_{(0,2)}$ respectively), both these can not be asserted together when the $\fut_{(1,2)}$ modality is within the scope of an unbounded ($\fut_{(0,\infty)}$) modality. We shall prove this by showing that here is no \mitlfpzinf\/ formula for $\phi$, using an \mitlfpzinf\/ EF game (as described below).

Let $m=(n+1)(k+1)$. Consider two families of words $\aaa_{n,k}$ and $\bbb_{n,k}$ with $untime(\aaa_{n,k}) = untime(\bbb_{n,k}) = a (ac)^{2m+1}$. Both families of words have all events except the initial $a$, occurring beyond the timestamp $k+1$. Hence, all the letters are at a time distance in $(k,\infty)$ from the origin. The intuition behind this is to disallow the \ssp\/ to distinguish integer boundaries between events. 
We shall call each $ac$-pair a \emph{segment}. The words are depicted in Figure \ref{fig:expressgame2}. Let $\delta$ be such that $0<\delta << 1/(2m+1)^2$. Consider the word $\aaa_{n,k}$ such that the segments occur beyond $k+1$ as follows. A segment begins with the occurrence of an $a$ at some time stamp (say $x$), and has a $c$ occurring at $x+1-\delta$. The following segment begins at $(x+1+2\delta)$. Since each $a$ has a $c$ from its own segment within time distance $(1,2)$ time units and a $c$ from the following segment at a time distance $>2$ time units and successive $c$'s are separated by a time distance $>1$ time units, it is easy to verify that $\forall n,k ~.~ \aaa_{n,k}\not\models\phi$.

The timed word $\bbb_{n,k}$ is identical to $\aaa_{n,k}$ except for the positioning of $c$ in the $(m+2)^{nd}$ segment, which is at a time distance $2-\epsilon$ from the $a$ of the middle $(m+1)^{st}$ segment, for some $\epsilon<<\delta$. Due to this $c$ which is at a time distance within $(1,2)$ time units from the middle $a$, we can conclude that $\forall n,k ~.~ \bbb_{n,k}\models \phi$.

Now consider an $n$ round \mitlfpzinf- EF game, with $\maxint =k$. We shall say that the game in a given round is in \emph{identical configurations} if the initial configuration of the round is of the form $(i,i)$. For $1\leq x\leq (2m+1)$, denote the $a$ and $c$ of the $x^{th}$ segment by $a_x$ and $c_x$ respectively. Consider the following strategy for the \ddp.
\begin{itemize}
\item Starting from identical configurations, the \ddp\/ may mimic the \ssp's moves at all times and maintain identical configurations except in the following cases, which lead to non-identical configurations (these are depicted by dotted arrows in Figure \ref{fig:expressgame2}).
\begin{itemize}
\item Starting from the identical configuration $(a_{m+1},a_{m+1})$ in the middle segment, the \ssp\/ chooses the interval $(0,2)$ and chooses $c_{m+2}$ in $\bbb_{n,k}$ then the \ddp\/ must respond by choosing $c_{m+1}$ in $\aaa_{n,k}$. 
\item Starting from the identical configuration $(a_{m+1},a_{m+1})$ in the middle segment, the \ssp\/ chooses the interval $(2,\infty)$ and chooses $c_{m+2}$ in $\aaa_{n,k}$ then the \ddp\/ must respond by choosing $c_{m+3}$ in $\bbb_{n,k}$.
\item Starting from the identical configuration $(c_{m+1},c_{m+1})$, the \ssp\/ chooses the interval $(0,2)$ and chooses $a_{m+1}$ in $\bbb_{n,k}$ then the \ddp\/ must respond by choosing $a_{m+2}$ in $\aaa_{n,k}$.
\item Starting from the identical configuration $(c_{m+1},c_{m+1})$, the \ssp\/ chooses the interval $(2,\infty)$ and chooses $a_{m+1}$ in $\aaa_{n,k}$ then the \ddp\/ must respond by choosing $a_{m}$ in $\bbb_{n,k}$.
\end{itemize}
Note that the only position where the two words differ is in the timestamp of $c_{m+2}$ and the only position from which they can be differentiated by an integeral time distance is from $a_{m+1}$. Hence starting from identical configurations, the only way in which a non-identical configuration may be achieved is by one of the above possibilities and the resulting configuration has a segment difference of 1.
\item The segment difference of a non-identical configuration is $\geq ~1$. Starting from a non-identical configuration, we have the following possibilities
\begin{itemize}
\item If the initial configuration includes $a_{m+1}$ or $c_{m+2}$ in $\bbb_{n,k}$, then the \ssp\/ may choose a suitable interval and choose the position $c_{m+1}$ or ${a_m+1}$ respectively. This increases the segment difference of the configuration by 1. For example, if the initial configuration is $(c_{m+3},c_{m+2})$ (this may be achieved, for instance, when the \ssp\/ chooses the interval (0,1) for two successive ``future rounds'', starting from an initial non-identical configuration of $(a_{m+1},a_m)$), and the \ssp\/ chooses the interval (0,2) and chooses $a_{m+1}$ in $\bbb_{n,k}$ then the \ddp\/ is forced to choose $a_{m+2}$, resulting in a segment difference of 2.
\item For any other choice of positions by the \ssp\/, the \ddp\/ may ``copy'' the \ssp\/'s moves by moving exactly the number of positions that is moved by the \ssp. This will result in non-identical configuration with segment difference less than or equal to the initial segment difference in the round.
\end{itemize}
\end{itemize}
It is now easy to argue that the above strategy is a winning strategy for the \ddp\/ for an $n$-round \mitlfpzinf\/ game with $\maxint=k$. By observing the two words, we can see that the only way the \ssp\/ can win a round is by beginning with non-identical configurations at either end of the words (at time distance $\leq k$ from either the beginning or end of the sequence of segments), such that the \ssp\/ may have an $a$ or a $c$ to chooses in one of the words, while the \ddp\/ wouldn't. However, the first time a non-identical configuration is achieved, is in the middle segment (as discussed above). The \ssp\/ has two choices- (i) to increase the segment difference by repeatedly choosing the middle segment configurations, or (ii) to maintain the segment difference and move towards either end of the word. In order to maintain a segment difference $\geq 1$ and reach either end of the word, the \ssp\/ can move a maximum time distance of $k$ time units (if it chooses an interval larger than that, the \ddp\/ may be able to achieve identical configurations). Hence the \ssp\/ requires at least $n$ rounds to reach either end of the word whilst maintaining non-identical configurations. Since the game is only of $n$ rounds, the \ssp\/ will not have enough rounds to first establish a non-zero segment difference and maintain it, while traversing to either end of the words. Hence, the above strategy is a winning strategy for the \ddp\/ for an $n$-round \mitlfpzinf\/ EF-game with $\maxint=k$.
\end{proof}

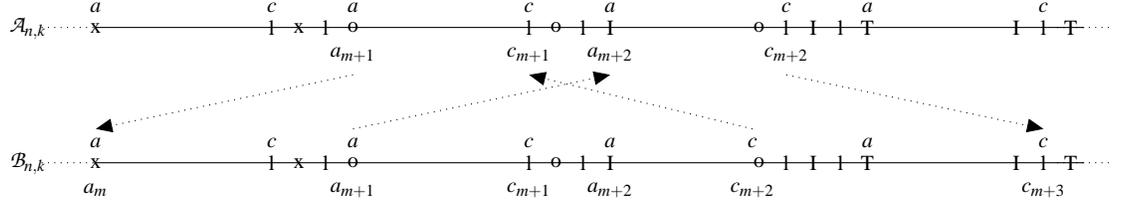
\begin{figure}
\begin{tikzpicture}[scale=0.9, transform shape]
\draw (0,4) node{$\aaa_{n,k}$};
\draw[dotted] (0.2,4)--(1,4); \draw[dotted] (15.2,4)--(16,4);
\draw(1,4) node{x}-- (3.6,4) node{l}-- (4,4) node{x}--(4.4,4) node{l}-- (4.8,4) node{o}-- (7.4,4) node{l} --(7.8,4) node{o}--(8.2,4) node{l}-- (8.6,4) node{I}-- (10.8,4) node{o}-- (11.2,4) node{l}-- (11.6,4) node{I}--(12,4) node{l}-- (12.4,4) node{T}-- (14.6,4) node{I}-- (15,4) node{l}-- (15.4,4) node{T} -- (15.6,4) node{};
\draw(1,4.3) node{$a$}; \draw (3.6,4.3) node{$c$}; \draw (4.8,4.3) node{$a$}; \draw  (7.4,4.3) node{$c$}; \draw (8.6,4.3) node{$a$}; \draw (11.2,4.3) node{$c$}; \draw (12.4,4.3) node{$a$}; \draw (15,4.3) node{$c$};

\draw (0,2) node{$\bbb_{n,k}$};
\draw[dotted] (0.2,2)--(1,2); \draw[dotted] (15.2,2)--(16,2);
\draw(1,2) node{x}-- (3.6,2) node{l}-- (4,2) node{x}--(4.4,2) node{l}-- (4.8,2) node{o}-- (7.4,2) node{l} --(7.8,2) node{o}--(8.2,2) node{l}-- (8.6,2) node{I}-- (10.8,2) node{o}-- (11.2,2) node{l}-- (11.6,2) node{I}--(12,2) node{l}-- (12.4,2) node{T}-- (14.6,2) node{I}-- (15,2) node{l}-- (15.4,2) node{T} -- (15.6,2) node{};
\draw(1,2.3) node{$a$}; \draw (3.6,2.3) node{$c$}; \draw (4.8,2.3) node{$a$}; \draw  (7.4,2.3) node{$c$}; \draw (8.6,2.3) node{$a$}; \draw (10.7,2.3) node{$c$}; \draw (12.4,2.3) node{$a$}; \draw (15,2.3) node{$c$};

\draw(4.8,3.6) node {$a_{m+1}$};
\draw(4.8,1.6) node {$a_{m+1}$};
\draw(8.6,3.6) node {$a_{m+2}$};
\draw(8.6,1.6) node {$a_{m+2}$};
\draw(11.2,3.6) node {$c_{m+2}$};
\draw(10.7,1.6) node {$c_{m+2}$};
\draw(7.4,3.6) node {$c_{m+1}$};
\draw(7.4,1.6) node {$c_{m+1}$};
\draw (1,1.6) node{$a_m$};
\draw (15,1.6) node{$c_{m+3}$};

\draw [arrows= -triangle 45,dotted] (10.7,2.5)--(7.4,3.3);
\draw [arrows= -triangle 45,dotted] (11.2,3.3)--(15,2.5);
\draw [arrows= -triangle 45,dotted] (4.8,2.5)--(8.6,3.3);
\draw [arrows= -triangle 45,dotted] (4.8,3.3)--(1,2.5);
\end{tikzpicture}
\caption{\mitlfpzinf\/ EF game to show that \mitlfpzinf $\subsetneq$ \mitlfp}
\label{fig:expressgame2}
\end{figure}

\begin{theorem}
\label{thm:fpexpress}
\mitlfpinf\/ $\subsetneq$ \mitlfpzinf
\end{theorem}
\begin{proof}
Logic \mitlfpinf\/ is a syntactic fragment of \mitlfpzinf. We shall now show that it is strictly less expressive than \mitlfpzinf\/ by showing that there is no \potdta\/ which accepts the language given by the formula $\phi:= \fut_{[0,\infty)}(a\land \fut_{[0,2)} c)$. The proof relies on the idea that since a \potdta\/ may be normalized to one that has a bounded number of clocks (bounded by the number of progress edges), and every edge may reset a clock at most once on a given run, the \potdta\/ cannot ``check'' every $a$  for its matching $c$ in a timed word which has sufficiently many $ac$ pairs.

Assuming to contrary, let $\autm$ be a \potdta\/ with $m$ number of progress edges , such that $\mathcal L(\phi)=\mathcal L(\autm)$. Now consider the word $\rho$ consisting of the event sequence $(ac)^{4m+1}$ where the $x^{th}$ $ac$ pair gives the timed subword $(a,3x)(c,3x+2.5)$. Thus, each $c$ is $2.5$ units away from the preceding $a$. Hence, $\rho\not\in\mathcal L(\phi)$. 
Consider the run of $\autm$ over $\rho$. 
There are a maximum number of $m$ clocks in $\autm$ that are reset, in the run over $\rho$.

By a counting argument (pigeon-hole principle), there are at least $m+1$ (possibly overlapping but distinct) subwords of $\rho$ of the form $acacac$, none of whose elements have been ``clocked'' by $\autm$. Call each such subword a group. Enumerate the groups sequentially.
Let $v_j$ be a word identical to $\rho$ except that the $j^{th}$ group is altered, such that its middle $c$ is shifted by 0.7 t.u. to the left, so that $v_j$ satisfies the property required in $\phi$. Note that there are at least $m+1$ such distinct $v_j$'s and for all $j$, $v_j\in\mathcal L(\phi)$. \\
Given a $v_j$, if there exists a progress edge $e$ of $\autm$ such that in the run of $\autm$ on $v_j$, $e$ is enabled on the altered $c$, then for all $k\neq j$, $e$ is not enabled on the altered $c$ of $v_k$. This is because, due to determinism, the altered $c$ in $v_j$ must satisfy a guard which neither of its two surrounding $c$'s in its group can satisfy. \\
From the above claim, we know that the $m$ clocks in $\autm$, may be clocked on at most $m$ of the altered words $v_j$. However, the family $\{v_j\}$ has at least $m+1$ members. Hence, there exists a $k$ such that the altered $c$ of $v_k$, (and the $k^{th}$ group) is not reachable 
by $\psi$ in $\rho$ or any of the $\{v_j\}$. Hence $w\models\psi$ iff $v_k\models \psi$. But this is
a contradiction as $\rho \notin L(\phi)$ and $v_k \in L(\phi)$ with $L(\phi)=L(\psi)$. \\
Therefore, there is no \potdta\/ which can express the language $\mathcal L(\phi)$.
\end{proof}

\section{Discussion}
We have shown how unaryness, coupled with timing restrictions, yields interesting fragments of \mitl, that are placed lower in the hierarchy, in terms of expressiveness and decidability complexities.
\begin{itemize}
\item We have a \emph{B\"uchi-like logic-automaton characterization} for a fragment of \mtl. \potdta\/ are closed under complementation, and have the key ingredients required, for convenient logic-automata conversions. Moreover, the exact characterization of \potdta, in terms of an \mtl\/ fragment (i.e. \mitlfpinf) has been identified.
\item The embedding of \mitlfpb\/ into \potdta, and the conversion from \potdta\/ to \mitlfpinf\/ formulas shows us how even properties like punctuality and boundedness, when coupled with determinism, may be expressed by a unary and non-punctual fragment, i.e. \mitlfpinf.
\item Unary fragments \mitlus\/ and \mitluszinf\/ do not result in a gain in satisfiability complexity and there is a significant gain (from \pspace-complete to \np-complete), when \mitlusinf\/ is compared to its unary-counterpart, \mitlfpinf.
\item While satisfiability of \mitlfz\/ is \pspace-complete, the satisfiability of \mitlfpinf\/ is \np-complete. This asymmetry in decision complexities of logics with one-sided constraints, and on finite words, has been observed.
\item An \nexptime-complete fragment of \mitl\/ has been recognized: \mitlfpb, which combines the restrictions of unaryness, punctuality, as well as boundedness is a rather restrictive logic in terms of expressiveness, but more succinct than \mitlfpinf.  In general, for the unary fragments of \mitl, one-sided interval constraints prove to yield much better decision complexities than bounded interval constraints.
\end{itemize}

In this chapter, our logic-automata reductions rely on strict monotonicity of the pointwise models. We believe that the results may be extended to the weakly-monotonic case, by using some concepts similar to the untimed case and still maintaining the decision complexities. 

\bibliography{mybib}

\end{document}

%% file: macros.tex
\newcommand{\fif}{\rmiff}

\newcommand{\rmiff}{\mathbin{~\mbox{iff}~}}

\newcommand{\class}{\textsc}

\newcommand{\defn}{\mathrel{\mbox{$~\stackrel{\rm def}{=}~$}}}

\newcommand{\until}{\textsf{U}}
\newcommand{\since}{\textsf{S}}
\newcommand{\fut}{\textsf{F}}
\newcommand{\past}{\textsf{P}}

\newcommand{\POTDFA}{\mbox{$\mathit{po2dfa}$}}

\newcommand{\oomit}[1]{}
\newcommand{\df}{\defn}

\newcommand{\rend}{\triangleleft}
\newcommand{\lend}{\triangleright}

\newcommand{\smartqed}

\newcommand{\mtlus}{\mbox{$\mathit{MTL[\until_I,\since_I]}$}}
\newcommand{\mtlfp}{\mbox{$\mathit{MTL[\fut_I,\past_I]}$}}
\newcommand{\mitlus}{\mbox{$\mathit{MITL[\until_I,\since_I]}$}}
\newcommand{\bmtlus}{\mbox{$\mathit{Bounded ~MTL[\until_b,\since_b]}$}}

\newcommand{\mitlfp}{\mbox{$\mathit{MITL[\fut_I,\past_I]}$}}

\newcommand{\bmitlfp}{\mbox{$\mathit{Bounded ~MITL[\fut_b,\past_b]}$}}
\newcommand{\tptlus}{\mbox{$\mathit{TPTL[\until,\since]}$}}

\newcommand{\potdta}{\mbox{$\mathit{po2DTA}$}}

\newcommand{\unite}{\otimes}

\newcommand{\autm}{\mathcal A}

\newcommand{\potdfa}{\POTDFA}
\newcommand{\twodta}{\mbox{$\mathit{2DTA}$}}

\newcommand{\utl}{\mbox{$\mathit{TL[F,P]}$}}

\newcommand{\ttl}{\mbox{$\mathit{TTL[X_\theta, Y_\theta]}$}}

\newcommand{\ssp}{\mbox{$\mathit{Spoiler}$}}
\newcommand{\ddp}{\mbox{$\mathit{Duplicator}$}}
\newcommand{\aaa}{\mathcal A}
\newcommand{\bbb}{\mathcal B}

\newcommand{\maxint}{\mbox{$\mathit{MaxInt}$}}

\newcommand{\ttlxy}{\mbox{$\mathit{TTL}[X_\theta,Y_\theta]$}}
\newcommand{\mtlu}{\mbox{$\mathit{MTL[\until_I]}$}}

\newcommand{\mitlf}{\mbox{$\mathit{MITL[\fut_I]}$}}

\newcommand{\mitlfpinf}{\mbox{$\mathit{MITL[\fut_\infty,\past_\infty]}$}}
\newcommand{\mitlusinf}{\mbox{$\mathit{MITL[\until_\infty,\since_\infty]}$}}
\newcommand{\mitlfpb}{\bmitlfp}
\newcommand{\mtlusb}{\bmtlus}
\newcommand{\mitlfb}{\mbox{$\mathit{Bounded ~MITL[\fut_b]}$}}

\newcommand{\mitl}{\mbox{$\mathit{MITL}$}}

\newcommand{\mtl}{\mbox{$\mathit{MTL}$}}
\newcommand{\ltlus}{\mbox{$\mathit{LTL[\until,\since]}$}}
\newcommand{\ltlfp}{\mbox{$\mathit{LTL[\fut,\past]}$}}
\newcommand{\mitluszinf}{\mbox{$\mathit{MITL[\until_{0,\infty},\since_{0,\infty}]}$}}
\newcommand{\mitlfpzinf}{\mbox{$\mathit{MITL[\fut_{0,\infty},\past_{0,\infty}]}$}}
\newcommand{\mitlfzinf}{\mbox{$\mathit{MITL[\fut_{0,\infty}]}$}}
\newcommand{\mitlfinf}{\mbox{$\mathit{MITL[\fut_{\infty}]}$}}
\newcommand{\mitlfpz}{\mbox{$\mathit{MITL[\fut_{0},\past_{0}]}$}}
\newcommand{\mitlfz}{\mbox{$\mathit{MITL[\fut_{0}]}$}}

\newcommand{\pspace}{\mbox{$\class{PSPACE}$}}
\newcommand{\expspace}{\mbox{$\class{EXPSPACE}$}}
\newcommand{\nexptime}{\mbox{$\class{NEXPTIME}$}}
\newcommand{\conp}{\mbox{$\class{CoNP}$}}
\newcommand{\np}{\mbox{$\class{NP}$}}
\newcommand{\nexp}{\mbox{$\class{NEXPTIME}$}}

\newcommand{\tul}{\mbox{$\mathit{TUL}$}}

\newcommand{\idx}{\mathit{Idx}}
\newcommand{\frst}{{\mathcal F}}
\newcommand{\lst}{{\mathcal L}}

\newcommand{\tab}{\hspace*{1cm}}

